\documentclass{article}

\usepackage[preprint, nonatbib]{neurips_2023}

\usepackage[utf8]{inputenc}
\usepackage[T1]{fontenc}

\usepackage{hyperref}
\usepackage{url}
\usepackage{graphicx}
\usepackage{float}
\usepackage{booktabs}
\usepackage{multirow}

\usepackage{tabularx}
\usepackage{nicefrac}
\usepackage{microtype}
\usepackage{ragged2e}

\usepackage{amsmath}
\usepackage{amsfonts}
\usepackage{mathtools}
\usepackage{amsthm}
\newtheorem{proposition}{Proposition}
\newtheorem{remark}{Remark}

\usepackage{xcolor}
\definecolor{mint}{HTML}{00FFFF}
\definecolor{magenta}{HTML}{ff00ff}
\definecolor{green}{HTML}{00ff00}
\usepackage{newunicodechar}
\newunicodechar{−}{-}

\usepackage{enumitem}
\usepackage{caption}
\usepackage{subcaption}
\usepackage{natbib}

\usepackage{tikz}
\usetikzlibrary{positioning}

\usepackage{listings}
\usepackage{tcolorbox}
\usepackage{algorithm}
\usepackage{algpseudocode}

\usepackage{cleveref}
\usepackage{textcomp}
\usepackage{etoolbox}

\usepackage{booktabs}
\usepackage{multirow}
\usepackage[table]{xcolor}

\definecolor{codegray}{rgb}{0.5,0.5,0.5}
\definecolor{codepurple}{rgb}{0.58,0,0.82}
\definecolor{backcolour}{rgb}{0.95,0.95,0.92}

\lstdefinestyle{ mystyle}{
    backgroundcolor=\color{backcolour},   
    commentstyle=\color{codegray},
    keywordstyle=\color{magenta},
    numberstyle=\tiny\color{codegray},
    stringstyle=\color{codepurple},
    basicstyle=\ttfamily\footnotesize,
    breakatwhitespace=false,         
    breaklines=true,                 
    captionpos=b,                    
    keepspaces=true,                 
    numbers=left,                    
    numbersep=5pt,                  
    showspaces=false,                
    showstringspaces=false,
    showtabs=false,                  
    tabsize=2
}
\pretocmd{\section}{\ifnum\value{section}>2 \clearpage\fi}{}{}

\title{SHIFT: Sigmoid-Based Heuristic Invertible Fitness-Landscape Transformation \\for Accelerating SBST\thanks{%
       Code and data available at: \url{https://github.com/Jeong-jin-Han/SHIFT-SBST.git}}}

\vspace{-2.5mm}
\author{
  Jeongjin Han$^{1}$ \quad
  Seunghoon Sim$^{1}$ \quad
  Jian Lee$^{1}$ \quad
  Seongyoon Park$^{1}$ \\[4pt]
  $^{1}$School of Computing, KAIST \\
  \texttt{hjj22@kaist.ac.kr}
}

\begin{document}
\maketitle

\vspace{-2.5mm}

\begin{abstract}
Search-Based Software Testing (SBST) automates test input generation but is frequently hindered by challenging fitness landscapes characterized by numerous deceptive local optima that impede search progress, as well as extended plateaus where informative fitness signals are scarce. To address this bottleneck, we propose SHIFT (Sigmoid-Based Heuristic Invertible Fitness-Landscape Transformation for Accelerating SBST), a method designed to compress local landscapes and facilitate escape from stagnant regions without altering global semantics. By systematically contracting dense regions where search points cluster, the approach preserves mapping invertibility while enabling optimization algorithms to traverse more effectively toward global coverage with the same step size. When evaluated against established baselines—pure hill climbing and genetic algorithms—under a normalized experimental protocol, the proposed technique yields consistent improvements in convergence speed and search efficiency. These results demonstrate that sigmoid compression constitutes a lightweight yet effective mechanism for achieving more reliable coverage discovery in complex testing environments.

\end{abstract}
\section{Introduction}

Software testing remains one of the most resource-intensive tasks in software engineering, particularly when manual test design must account for complex input spaces and diverse program behaviors. Search-Based Software Testing (SBST) aims to alleviate this burden by framing test generation as an optimization problem, employing metaheuristics to search for inputs that trigger new coverage automatically. Among these techniques, Hill Climbing (HC) is widely used due to its simplicity, low computational overhead, and strong performance on smooth fitness landscapes.

However, the efficacy of hill climbing degrades significantly when the underlying search space is rugged or full of plateaus. Real-world software often produces fitness landscapes characterized by cliffs, deceptive plateaus, and clusters of local optima, especially around branch conditions involving complex logic, multi-variable interactions, or discontinuous boundary checks. Once the search falls into one of these regions, local fitness updates vanish, and the algorithm repeatedly revisits neighbors that yield no meaningful progress. In practice, this results in prolonged stagnation phases, wasted evaluation budgets, and an overall failure to reach deeper, hard-to-cover branches.

These limitations expose a structural mismatch: hill climbing is inherently myopic and performs best on landscapes where unproductive local regions are sparse and global transitions are accessible. Rather than attempting to render the hill climbing algorithm more sophisticated through complex operator variations or surrogate fitness estimators, we instead adopt a more fundamental approach: directly reshaping the geometry of the search space itself. This paper introduces \emph{SHIFT} (Sigmoid-Based Heuristic Invertible Fitness-Landscape Transformation for Accelerating SBST), a framework that compresses the fitness landscape to address this mismatch. By modifying the landscape geometry, we enable standard hill climbing to operate more effectively using its existing mechanics, without altering the global semantics of the test input domain.

Compressing local regions using a smooth, invertible sigmoid transformation serves to narrow flat basins, reduce the dominance of small-scale local optima, and "stretch" transitions toward more promising areas. Consequently, the search is less likely to become trapped early and obtains a clearer gradient toward global coverage-relevant zones. Unlike discrete mutation adjustments, this transformation uniformly reshapes the region around the current search point, enabling more meaningful exploratory steps without significantly increasing computational overhead.

The remainder of this paper presents the motivation for the SHIFT model, describes the mathematical design of the sigmoid-based compression that forms the core of our approach, and evaluates its impact on coverage efficiency by comparing its performance against standard hill climbing and a genetic algorithm as baselines. The experimental results demonstrate that compressing the solution space constitutes a lightweight, practical, and effective strategy for recovering hill climbing’s effectiveness in rugged SBST scenarios.

This paper makes the following contributions.
We propose SHIFT, an invertible sigmoid-based fitness-landscape transformation that compresses flat and near-flat basins while maintaining a strict one-to-one correspondence with the original search space.
Unlike prior smoothing approaches that alter the global structure of the landscape, SHIFT operates through a provably bijective coordinate change: we show formally that the transformation cannot displace or duplicate any global optimum (Proposition~\ref{prop:global-opt}), providing a theoretical foundation that most heuristic landscape-shaping methods lack.
To realize SHIFT in practice, we further introduce a bidirectional basin detection algorithm and a dimension-aware compression manager that identify and reshape the search space adaptively across independent input dimensions, reducing the effective per-iteration cost from $O(n)$ to $O(k)$ over the $k \ll n$ remaining active dimensions (Section~\ref{sec:compression}).
Finally, we present a comprehensive empirical evaluation of HC-SHIFT against standard hill climbing and a genetic algorithm on a suite of synthetic and real SBST benchmarks under identical time budgets, demonstrating consistent gains in branch coverage, convergence speed, and robustness across a broad range of structurally challenging fitness landscapes (Section~\ref{sec:experiment}).

\newpage
\section{Related Works}

Research in Search-Based Software Testing (SBST) has explored various strategies to mitigate the difficulties posed by rugged or deceptive fitness landscapes. Although prior work does not provide an invertible transformation directly comparable to our SHIFT, several studies attempt to reshape or reinterpret the landscape in ways that improve search effectiveness.

A notable example is the \textit{surF} methodology proposed by \cite{manzoni2020surfing}, which applies a discrete Fourier transform and filters out high-frequency components, such as sharp spikes, to smooth the fitness landscape. By suppressing these high-frequency signals, \textit{surF} constructs a smoother surrogate landscape while preserving the location of the global optimum. This approach demonstrates that smoothing or reshaping the evaluation function can make local search algorithms more effective, even though the surrogate itself is not strictly invertible back to the original domain. Their work highlights that modifying the landscape can reduce the number of local optima and help search algorithms escape deceptive regions, supporting the conceptual motivation behind our compressive mapping.

Comparative studies between search heuristics in SBST also shed light on why hill climbing benefits from such landscape shaping. Earlier foundational work by \cite{harman2010theoretical} compared hill climbing with multi-population GAs for structural test generation. They identified problem classes where GAs outperform HC, particularly when the fitness structure resembles a “Royal Road” landscape with decomposable building blocks, where the crossover operator can significantly benefit. However, they also showed that HC can be competitive or even superior when the landscape is fragmented or discontinuous. This further motivates the need to reshape the landscape for HC rather than rely solely on algorithmic complexity.

Tabu Search has also been applied to test-case generation, primarily as a 
mechanism to avoid cycling and escape shallow local optima. \cite{ma2022scalable} demonstrate a tabu-based path exploration strategy that maintains a memory of recently visited test inputs to prevent re-evaluating unproductive regions. While their approach does not perform dimension-wise tabu operations—which aligns more closely with our compression model—it illustrates that memory-based mechanisms can help push search beyond local traps. However, the tabu list in existing SBST work is applied to complete test inputs rather than individual parameter dimensions, and the literature presents no known instance of tabu applied at the per-axis granularity. This absence highlights a gap that our SHIFT model partially addresses: instead of forbidding individual points or parameter values, our method geometrically reshapes the search space to make problematic regions less dominant and reduce the need for explicit 
memory.

Overall, prior work converges on a common theme: when the fitness landscape is rugged or poorly structured, algorithmic improvements alone seldom yield substantial gains. Whether through Fourier smoothing, empirical comparisons of HC and GA, or tabu-based memory mechanisms, the literature consistently points to landscape structure as the core limiting factor. Our work builds on this insight by providing a direct, invertible transformation that compresses local regions and exposes broader global transitions, enabling simple hill climbing to recover its effectiveness without reliance on heavy evolutionary machinery or complex hyperparameter tuning.

\newpage



\section{Methodology}
\label{sec:approach}

\subsection{Terminology and the Proposed SHIFT Framework}

For clarity, we introduce the terminology used throughout this paper when 
contrasting our approach with baseline search algorithms. At the core of our 
method lies \textbf{SHIFT}, a sigmoid-based, invertible transformation framework 
designed to reshape the fitness landscape encountered in SBST. SHIFT applies 
axis-aligned warpings that compress locally flat or weakly varying basins—regions 
in which classical local search becomes trapped—while preserving a 
one-to-one correspondence between the original and transformed search spaces. 
Because these transformations remain lightweight and invertible, the framework 
can be integrated into existing metaheuristics without altering their algorithmic 
structure.

When combined with standard hill climbing, SHIFT yields the proposed method 
\textbf{HC-SHIFT}. The underlying hill-climbing procedure is preserved, but the 
added transformation enables the search to escape local minima in a controlled 
manner by selectively compressing basins along dimensions identified as relevant 
to fitness improvement.

For comparison, we refer to two baselines: \textbf{HC}, denoting the 
unmodified hill-climbing algorithm, and \textbf{GA}, denoting the genetic 
algorithm without any landscape transformation. Throughout the remainder of this 
report, these terms are used consistently across tables, figures, and 
experimental discussions; unless stated otherwise, HC and GA indicate the 
non-transformed baselines, while HC-SHIFT refers to our proposed approach.

\subsection{Program Instrumentation and Fitness Evaluation}
\label{sec:instrumentation}

Our framework follows the standard Search-Based Software Testing (SBST) 
pipeline: (1) instrument the program under test (PUT), (2) collect branch 
execution information, and (3) compute fitness for each input candidate.  
Instrumentation is performed via a Python source-to-source transformation that 
injects branch probes directly into the Abstract Syntax Tree (AST) of the PUT. 
All instrumentation and fitness logic is implemented in \texttt{sbst\_core.py}.

\paragraph{AST-based Instrumentation.}
\mbox{}
\vspace{0.5em}\\
Given a Python function, we parse its source into an AST and apply a transformer 
that rewrites each Boolean expression into a call to a branch-probe runtime 
(\texttt{BranchProbe}). Each conditional construct—including \texttt{if}, 
\texttt{while}, \texttt{for}, pattern matching, chained comparisons, and logical 
connectives—is assigned a unique branch identifier (\texttt{bid}). During 
execution, the probe records both the Boolean outcome and the tightest observed 
branch-distance values for the True and False directions, producing a complete 
mapping 
\[
    \text{bid} \;\mapsto\; 
        \{\texttt{outcome},\, d_{\text{true}},\, d_{\text{false}}\}
\]
for each evaluated input.  
Additionally, for every branch we compute a \emph{guard chain}, i.e., the list 
of ancestor branch conditions that must be satisfied before the target branch 
can be reached. These guard chains, collected during AST traversal, are later 
used to compute the approach level.

\paragraph{Branch-Distance Computation.}
\mbox{}
\vspace{0.5em}\\
For elementary comparisons such as $a < b$ or $x == y$, we follow the standard 
SBST branch-distance formulation. Given operands $a$, $b$ and operator 
$\mathit{op}$, we define
\[
    \mathit{raw\_f}(a,b,\mathit{op}) = 
    \begin{cases}
        (a-b)+1,& a < b, \\
        (b-a)+1,& a > b, \\
        |a-b|,& a=b, \\
        -|a-b|,& a \ne b,
    \end{cases}
\]
which is then mapped to a non-negative branch distance 
$d \in [0,\infty)$ using operator-specific rules (with $d=0$ indicating that 
the condition is satisfied).  
These rules are implemented in \texttt{raw\_f}, \texttt{bd}, and 
\texttt{b\_from\_raw}.  
The framework further extends branch-distance computation to element-wise 
comparison of tuples or lists, Boolean operators (\texttt{and}, \texttt{or}, 
\texttt{not}), membership expressions via interval compression 
(\texttt{in}/\texttt{not in}), and pattern-matching constructs by conditional 
lifting.

\newpage
\paragraph{Approach Level (AL).}
\mbox{}
\vspace{0.5em}\\
Let $\mathcal{G} = [(g_1,v_1),\dots,(g_k,v_k)]$ denote the ordered list of 
ancestor guards required to reach branch \(\text{bid}\), where each $v_i$ 
is either \text{True} or \text{False}.  
If $i^\ast$ is the index of the first unsatisfied guard, then the approach 
level is defined as
\[
    AL(\text{bid}) = |\mathcal{G}| - i^\ast,
\]
with $AL=0$ only if all guards are satisfied.

\paragraph{Normalised Branch Distance (nBD).}
\mbox{}
\vspace{0.5em}\\
Once the failing guard (or the target guard itself) is identified, its raw 
distance $d$ is normalized using
\[
    \mathrm{nBD}(d) = 1 - 1.001^{-d},
\]
which yields a smooth value in $[0,1)$ that decreases monotonically as the 
input approaches the required branch direction.

\paragraph{Final Fitness Function.}
\mbox{}
\vspace{0.5em}\\
The overall fitness for branch \(b\) and desired direction 
\(\texttt{want\_true} \in \{0,1\}\) is computed as
\[
    F = AL + \mathrm{nBD},
\]
with $F=0$ indicating that the branch was executed in the desired direction.  
This behavior is implemented in \texttt{fitness\_AL()}, which manages guard 
handling, approach-level computation, and branch-distance aggregation.

\paragraph{Execution and Trace Collection.}
\mbox{}
\vspace{0.5em}\\
The instrumented PUT is executed in an isolated probe environment. Before each 
execution the probe resets its internal state, and during execution it records 
branch outcomes and branch-distance values at every conditional evaluation.  
The collected trace is then passed directly into the fitness computation without 
the need for additional runtime instrumentation.

\paragraph{Handling Unreachable Loop Branches.}
\mbox{}
\vspace{0.5em}\\
Certain branch directions are semantically unreachable under Python's loop 
semantics and are therefore excluded during target generation.  
For \texttt{while True} loops, the False direction—representing loop exit—cannot 
occur unless an explicit \texttt{break} is present. The AST transformer detects 
constant-True loop conditions (line~1140 of \texttt{sbst\_core.py}) and marks 
such branches accordingly (lines~1152--1153).  
During target enumeration, \texttt{make\_targets\_for\_func()} 
( lines~1965--1968) automatically omits the unreachable False direction.

Similarly, for \texttt{for}-loops the False branch corresponds to the loop body 
not executing even once. Using iteration metadata tracked through 
\texttt{loop\_minlen} (line~544), the framework determines cases where the 
iterable is guaranteed to be nonempty; \texttt{make\_targets\_for\_func()} 
(lines~1970--1978) omits such False-direction targets as they are 
unsatisfiable.  

By filtering out unreachable loop branches, the framework avoids expending 
search effort on infeasible targets and ensures that fitness evaluation focuses 
solely on semantically meaningful branches.

\newpage
\subsection{Search Algorithms}
    \subsubsection{Baseline: Simple Hill Climbing (HC)}
\label{sec:hc}

As a baseline local search method, we employ a deterministic,
coordinate-wise hill climbing algorithm, hereafter referred to as
\textbf{HC}.  
This baseline operates directly on the original fitness landscape
introduced in Section~\ref{sec:instrumentation}, where each candidate input is
represented as an $n$-dimensional integer vector and evaluated using the
fitness function $F = AL + \mathrm{nBD}$.  
The full implementation is provided in
\texttt{hill\_climb\_simple\_nd\_code} (lines 827--902 of
\texttt{hill\_climb\_multiD.py}).


\paragraph{Overview.}
Simple hill climbing (HC) serves as the baseline local-search method in our
evaluation. The algorithm performs deterministic coordinate-wise descent on the
original fitness landscape, evaluating the $2n$ axis-aligned neighbors obtained
by adding or subtracting one unit in a single dimension. At each step, HC
selects the neighbor with the lowest fitness and moves only when an improvement
is found; otherwise, it terminates at a local minimum. The procedure operates
under the same per-branch time budget of $T=20$ seconds used for GA and
HC-SHIFT, and also obeys a maximum step limit of $K = 2000$. Because HC explores
neither diagonal directions nor alternative basins, and lacks any mechanism for
escaping plateaus or rugged regions where multiple neighbors share identical
fitness, it serves as a minimal but informative baseline against which the
behavior of GA and HC-SHIFT can be contrasted.

\paragraph{(a) Neighbor Generation and Search Strategy.}
\mbox{}
\vspace{0.5em}\\
At each iteration, HC considers exactly $2n$ axis-aligned neighbors obtained
by applying a fixed $\pm 1$ perturbation to a single coordinate:
\[
    \mathcal{N} = 
    \bigl\{
        (x_1,\ldots,x_d - 1,\ldots,x_n),\ 
        (x_1,\ldots,x_d + 1,\ldots,x_n)
    \bigr\}_{d=1}^{n}.
\]
No diagonal moves, adaptive step sizes, or multi-coordinate perturbations are
used; the search progresses solely through these minimal coordinate-wise
increments (lines 868--888).  
Because all $2n$ neighbors may share identical fitness values on plateaus, HC
has no mechanism to escape such regions and is therefore highly sensitive to
flat or rugged landscapes.

Among all evaluated neighbors, HC selects the candidate with the lowest
fitness (steepest descent, lines 891--894):
\[
    x' = \arg\min_{y \in \mathcal{N}} F(y).
\]
If $F(x') < F(x)$, the search moves to $x'$ and continues; otherwise it
terminates at a local minimum (lines 896--900).

\paragraph{(b) Time-Budget Enforcement.}
\mbox{}
\vspace{0.5em}\\
Like GA (Section~\ref{sec:ga}), our HC implementation enforces a strict
per-branch time budget to ensure fair comparison across all methods.  
Before evaluating each neighbor---and before the initial evaluation---HC
checks whether the elapsed time has exceeded the budget (lines 853--856,
862--865, 870--873, 880--883):
\[
    \text{if } (t_{\mathrm{current}} - t_{\mathrm{start}}) \ge t_{\max}
    \quad \Rightarrow \quad \text{terminate and return best}.
\]
If the time limit is reached at any point, HC immediately returns the current
best solution, even if the step limit has not been exhausted.

\paragraph{(c) Stopping Conditions.}
\mbox{}
\vspace{0.5em}\\
HC terminates under three circumstances.  
First, if none of the $2n$ neighbors yields a strictly lower fitness value, the
algorithm concludes that it has reached a local minimum and halts
(line~900).  
Second, the search stops when the maximum step limit of $K = 2000$ iterations
is reached (line~861).  
Finally, HC enforces a strict per-branch time budget and immediately
terminates once the allotted time has elapsed.  
In all termination cases, the algorithm returns the best point encountered
up to that moment, represented as the trajectory
$\{(x_0, F(x_0)), (x_1, F(x_1)), \ldots\}$ (lines~859, 898).

    \newpage
    \subsubsection{Baseline: Genetic Algorithm (GA)}
\label{sec:ga}

As a second baseline, we employ a standard genetic algorithm, hereafter
referred to as \textbf{GA}.  
The GA operates over the same $n$-dimensional integer input domain and uses
the fitness function $F = AL + \mathrm{nBD}$ defined in
Section~\ref{sec:instrumentation}.  
Our implementation follows a steady-state evolutionary scheme with
tournament selection, uniform crossover, independent per-gene mutation, and
elitist survivor selection.  
The full implementation is provided in \texttt{ga.py} (lines 20--246).
Importantly, no SHIFT-based transformations are applied---GA serves purely as
a population-based global-search method for comparison.

\paragraph{Overview.}
The genetic algorithm (GA) follows a standard steady-state evolutionary
procedure that iteratively refines a population of candidate inputs. The
algorithm begins by generating an initial population using either random or
biased initialization and evaluating all individuals under the same
per-branch time budget applied to HC and HC-SHIFT. Each generation then
consists of tournament-based parent selection, uniform crossover, and
independent per-gene mutation, with at least one mutation enforced to avoid
genetic stagnation. Elitist survivor selection preserves the top-performing
individuals, while deduplication removes redundant genotypes and helps
maintain diversity on flat landscapes. At every stage of evaluation—including
initialization, the start of each generation, and during offspring
evaluation—the GA enforces early stopping: if any individual achieves
fitness $F = 0$, the algorithm terminates immediately. Otherwise, evolution
continues until the time budget is exhausted, no further improvements occur,
or the generation limit is reached, at which point the best individual
encountered is returned.

\paragraph{(a) Representation and Initialization.}
\mbox{}
\vspace{0.5em}\\
Each candidate input is represented as an integer vector
$x \in \mathbb{Z}^n$ constrained by the automatically inferred bounds
$[L, H]$ of the instrumented program.  

The initial population of size $P = 10000$ is generated using either
\textbf{random} or \textbf{biased initialization}
(see Section~\ref{sec:init} for details).
Our experiments use biased initialization by default, sampling near
program constants to accelerate early convergence when control-flow
predicates depend on specific literals (e.g., \texttt{if x > 100}).
\paragraph{(b) Selection, Crossover, and Mutation.}
\mbox{}
\vspace{0.5em}\\
\textbf{Selection.}  
We use $k$-way tournament selection with $k=3$ to choose parents
(lines 103--106).  
For each parent, $k$ individuals are sampled uniformly at random from the
current population, and the one with the lowest fitness is selected.

\textbf{Crossover.}  
Given two selected parents $p_1, p_2 \in \mathbb{Z}^n$, uniform crossover
produces an offspring $c \in \mathbb{Z}^n$ by independently choosing each
gene from one parent or the other with equal probability:
\[
    c_i =
    \begin{cases}
        (p_1)_i & \text{with prob.\ } 0.5,\\
        (p_2)_i & \text{with prob.\ } 0.5,
    \end{cases}
\]
for every dimension $i \in \{1,\dots,n\}$ (lines 108--110).

\textbf{Mutation.}  
Mutation is applied independently to each gene with probability
$p_{\mathrm{mut}} = 1/n$ (lines 112--123).  
When a gene is mutated, it is perturbed by a randomly chosen integer step
from the set $\{-3, -2, -1, +1, +2, +3\}$ and clipped to remain within
$[\ell, u]$:
\[
    x_i' = \mathrm{clip}\bigl(x_i + \Delta\bigr),
    \qquad
    \Delta \in \{-3, -2, -1, +1, +2, +3\}.
\]
To prevent genetic stagnation on flat fitness regions, we enforce at least
one mutation per offspring (\texttt{ensure\_mutation=True}, lines 120--122):
if no gene was mutated during the per-gene pass, we forcibly mutate a
randomly selected gene.

\newpage
\paragraph{(c) Elitism and Diversity Control.}
\mbox{}
\vspace{0.5em}\\
After fitness evaluation, the population is sorted by fitness, and the top
$\lceil P \cdot \alpha_e \rceil$ individuals (elite ratio
$\alpha_e = 0.1$, i.e., 10\% of the population) are preserved verbatim into
the next generation (lines 183--185).  
This elitist strategy guarantees monotonic improvement of the best solution
across generations.

To maintain population diversity and reduce redundant fitness evaluations,
all duplicate individuals are removed via deduplication before the next
evaluation cycle (lines 125--131, 200).  
This is particularly important on flat landscapes where many candidates may
converge to the same point.

\paragraph{(d) Time-Budget Enforcement and Early Stopping.}
\mbox{}
\vspace{0.5em}\\
Unlike fixed-generation genetic algorithms, our implementation enforces a 
strict per-branch time limit identical to that used by HC 
(Section~\ref{sec:hc}) and HC-SHIFT (Section~\ref{sec:compression}). Before 
evaluating any candidate—both during initialization and within each 
generation—the GA checks whether the elapsed time has exceeded the
allowed budget (lines~143--154, 206--212):
\[
    \text{if } (t_{\mathrm{current}} - t_{\mathrm{start}}) \ge t_{\max}
    \quad \Rightarrow \quad \text{terminate and return best}.
\]
This mechanism guarantees that all three methods (HC, GA, HC-SHIFT) operate 
under identical computational constraints and that their performance can be 
compared fairly.

In addition to time-budget enforcement, the GA incorporates an early-stopping 
criterion: once any individual achieves fitness $F = 0$, indicating that the 
target branch has been exercised in the desired direction, the algorithm 
terminates immediately. Early stopping is applied consistently during the 
initial evaluation of the population (lines~167--170), at the beginning of 
each generation (lines~188--191), and during the evaluation of newly generated 
offspring (lines~229--233). In all such cases, the GA returns the discovered 
solution without expending the remaining time or generation budget.


    \newpage
    \subsubsection{Proposed: HC-SHIFT (N-Dimensional Compression Hill Climbing)}
\label{sec:compression}

Our main contribution is an $n$-dimensional hill climbing algorithm that
dynamically reshapes the search landscape by detecting and compressing
one-dimensional fitness basins.  
Rather than navigating the original integer-valued space directly, the
algorithm constructs an auxiliary nonlinear coordinate system in which
large flat regions are contracted into short intervals.  
This transformation enables efficient escape from plateaus and
slow-gradient regions that frequently arise in SBST fitness landscapes.

We refer to this sigmoid-based, invertible transformation framework as
\textbf{SHIFT} (Sigmoid-Based Heuristic Invertible
Fitness-Landscape Transformation  for Accelerating
SBST), and to the resulting algorithm as
\textbf{HC-SHIFT}.  
The complete implementation is available in \texttt{compression\_hc.py}.

\paragraph{Overview.}
The $(AL + \mathrm{nBD})$ fitness formulation (Section~\ref{sec:instrumentation}) 
frequently produces extended plateaus and broad basins in which many consecutive 
inputs share identical fitness. Classical hill climbing (HC, 
Section~\ref{sec:hc}) becomes trapped in such regions because all neighboring 
points appear equally promising. Techniques such as tabu search or simulated 
annealing escape through randomness or memory, yet they leave the underlying 
fitness geometry unchanged. HC-SHIFT takes a different approach: it explicitly 
reshapes the search space by identifying the extent of flat or weakly varying 
regions and compressing them into a smaller geometric footprint, thereby 
introducing new gradients that guide search progress.

Algorithm~\ref{alg:hc-shift-main} summarizes the overall procedure. Like HC and 
GA (Section~\ref{sec:ga}), HC-SHIFT adheres to a strict per-branch time budget 
of $T = 20$ seconds (Section~\ref{sec:timebudget}); before each fitness 
evaluation, the algorithm checks for timeout and terminates early if necessary 
to ensure fair comparison. The algorithm begins by marking all dimensions as 
active and initializing an empty compression manager. It then enters the main 
hill-climbing loop (Algorithm~\ref{alg:hill-climb-loop}), where it generates 
axis-aligned and diagonal neighbors in compressed space 
(Algorithm~\ref{alg:generate-neighbors}), tracks which dimensions contribute to 
fitness changes, and deactivates stagnant dimensions after $p$ consecutive 
steps.

When the hill-climbing loop can no longer make progress, HC-SHIFT performs 
bidirectional basin detection on each active dimension 
(Algorithm~\ref{alg:basin-detection}). These probes estimate 1D basin 
boundaries, update compression metadata, and construct SHIFT warpings 
$w_d$ that collapse flat regions into shorter intervals in $Z$-space. 
Finally, the algorithm evaluates candidates at the detected basin boundaries and 
restarts from the best point (Algorithm~\ref{alg:restart-selection}). This cycle 
repeats until convergence ($F=0$), the set of active dimensions becomes empty, 
or the time budget expires.

By combining fitness-driven basin detection, invertible space compression, 
dimension-level pruning, and systematic boundary restarts, HC-SHIFT provides a 
cohesive landscape-shaping framework that enables efficient navigation of 
complex $n$-dimensional search spaces.

\paragraph{(a) 1D Basin Detection.}
\mbox{}
\vspace{0.5em}\\
Given a local minimum $x^\ast$ along dimension~$d$, HC-SHIFT performs a 
bidirectional probing procedure to identify a contiguous interval
\(
    B_d = [\ell_d, r_d]
\)
representing the flat or near-flat basin surrounding $x^\ast$.  
Starting from $x^\ast$, the algorithm evaluates points one step at a time in 
both directions and classifies each point according to the following rules:

\begin{itemize}[nosep]
    \item \textbf{Rule~1 (Equal fitness).}  
          The fitness matches $F(x^\ast)$; the point is recorded as part of the 
          basin, and probing continues.
    \item \textbf{Rule~2 (Worse fitness).}  
          The fitness is worse but not better than $F(x^\ast)$; the point is 
          ignored but probing continues in search of a possible exit.
    \item \textbf{Rule~3 (Better fitness).}  
          The fitness strictly improves; probing stops immediately in that 
          direction, and the point is marked as a basin boundary.
\end{itemize}

During scanning, HC-SHIFT applies a priority policy for determining the basin 
boundary.  
If a Rule~3 point is encountered, its location is taken as the boundary for that 
direction.  
If no such point appears, the boundary is assigned to the farthest Rule~1 point, 
representing the maximal extent of the flat region.  
If neither Rule~1 nor Rule~3 occurs before the probing limit is reached—i.e., 
only Rule~2 points are observed—the algorithm concludes that no meaningful basin 
exists on that side.

\newpage
To prevent unbounded scanning on long plateaus or oscillatory regions, the 
probing distance is capped by a dimension-specific limit.  
The full detection logic is implemented in 
\texttt{detect\_compression\_basin()} (lines~126--224 of 
\texttt{compression\_hc.py}) and is invoked by 
Algorithm~\ref{alg:basin-detection}.  
If the resulting interval satisfies $|B_d| < 2$, HC-SHIFT skips compression for 
dimension~$d$ because the detected region is too small to constitute a 
meaningful basin.

\paragraph{(b) Sigmoid Warping (X-space $\to$ Z-space).}
\mbox{}
\vspace{0.5em}\\
For each detected basin $B_d=[s, s+L-1]$ of length $L$, HC-SHIFT builds
a smooth warping function $w_d: \mathbb{R}\to \mathbb{R}$ based on a
centered sigmoid:
\[
    z = w_d(x) = 
    \begin{cases}
        x & \text{if } x < s, \\
        s + \sigma\!\left(\alpha\left(\frac{x-s}{L}-\frac12\right)\right) 
            & \text{if } s \le x \le s+L, \\
        x - (L-1) & \text{if } x > s+L,
    \end{cases}
\]
where $\alpha=5.0$ controls the steepness and $\sigma(t)=1/(1+e^{-t})$ is
the standard sigmoid function.  
This mapping contracts the full basin $B_d$ of length $L$ into roughly
one unit in $z$-space while keeping the exterior regions affine
(shifted by $L-1$ on the right to maintain continuity).  
The inverse function $w_d^{-1}$ uses the logit to map compressed
coordinates back into valid integer positions in the original $X$-space.  
These transformations are implemented by the \texttt{SigmoidWarping}
class (lines 15--74) and constitute the core of the SHIFT framework.

\paragraph{Theoretical Guarantee.}
\mbox{}
\vspace{0.5em}\\
\begin{proposition}[Global Optimum Preservation]
\label{prop:global-opt}
Let $F: \mathcal{X} \to \mathbb{R}$ be a fitness function with global minimizer
$x^{*} \in \mathcal{X}$.
Let $w: \mathcal{X} \to \mathcal{Z}$ be the SHIFT transformation composed of
per-dimension warpings $w_d$ as defined above.
Then $w$ is a bijection, and the transformed fitness
$\tilde{F}(z) := F(w^{-1}(z))$
satisfies
\[
    \arg\min_{z \in \mathcal{Z}}\,\tilde{F}(z) \;=\; w(x^{*}).
\]
That is, the global minimizer is neither displaced nor duplicated by the SHIFT transformation.
\end{proposition}

\begin{proof}
Each per-dimension warping $w_d$ is continuous and strictly monotone: it is
affine (identity-shifted) outside the basin $B_d$ and strictly increasing
on $B_d$ because $\sigma' > 0$ everywhere.
Strict monotonicity implies injectivity; surjectivity onto the codomain
follows by continuity and the fact that $w_d$ covers all values outside
$B_d$ identically.
Hence each $w_d$ is a bijection, and their composition $w$ is also a
bijection.
Because $w$ is a bijection, $\tilde{F}(z) = F(w^{-1}(z))$ is merely a
reparametrization of $F$: every value attained by $F$ is attained by
$\tilde{F}$ exactly once at the corresponding transformed coordinate.
In particular, $\tilde{F}(w(x^{*})) = F(x^{*}) \le F(x) = \tilde{F}(w(x))$
for all $x \in \mathcal{X}$, so $w(x^{*})$ is the unique global minimizer of $\tilde{F}$.
\end{proof}

\begin{remark}
Proposition~\ref{prop:global-opt} is stated for the idealized continuous reparameterization underlying SHIFT.
In practice, the implementation operates on integer-valued input domains and recovers original coordinates via an integer-rounding inverse mapping, which constitutes an operational approximation of the exact bijection.
This approximation does not change which optimum the algorithm targets, but readers should note that the proof assumes the continuous formulation.
\end{remark}

\paragraph{Justification for the Sigmoid Choice.}
\mbox{}
\vspace{0.5em}\\
The sigmoid $\sigma(t) = 1/(1+e^{-t})$ is chosen over alternative compression
functions (e.g., piecewise-linear, tanh-based, or Gaussian kernels) for three
reasons.
First, it is \emph{smooth} ($C^\infty$), so the warped landscape contains no
slope discontinuities that would introduce artificial local optima at basin
boundaries—a critical property for gradient-following local search.
Second, it is \emph{strictly monotone}, which is the property exploited in
Proposition~\ref{prop:global-opt} to guarantee bijectivity and global optimum
preservation.
Third, its \emph{sigmoidal saturation} at both ends concentrates the
compression in the interior of the basin while leaving the exterior regions
only affine-shifted; in contrast, a piecewise-linear map would achieve
compression but introduce kinks, and a Gaussian-based kernel is not easily
invertible in closed form.
The steepness parameter $\alpha = 5.0$ controls the compression ratio: larger
values push more of the basin length into a narrower $Z$-space interval.
This value was selected so that a basin of moderate length is contracted to
roughly one unit in $Z$-space, making a single step in the compressed
coordinate equivalent to traversing the entire plateau in the original space.

\paragraph{(c) Metadata in Original X-Space.}
\mbox{}
\vspace{0.5em}\\
All compression metadata—specifically the basin boundaries $(s, L)$—is stored
directly in the original $X$-space rather than in the transformed $Z$-space.
Maintaining metadata in $X$-space ensures that all compression decisions remain
aligned with the true fitness function $F$, prevents inconsistencies that could
arise from repeatedly applying warpings, and supports stable behavior across
multiple compression cycles.  
This design also enables conflict-free merging of overlapping basins through
the routine \texttt{merge\_overlapping\_compressions} (lines~231--260), which
consolidates adjacent or intersecting compressed regions into a coherent
representation.  
The responsibilities associated with managing these structures are handled by
\texttt{MetadataCompressionOriginalSpace} (lines~76--121) together with
\texttt{CompressionManagerND} (lines~267--317), which jointly maintain the
state of all SHIFT-related metadata across dimensions.

\paragraph{(d) Multi-D Compression Manager and Active Dimensions.}
\mbox{}
\vspace{0.5em}\\
To extend basin compression to the multi-dimensional setting, HC-SHIFT
maintains a dynamically updated set of \emph{active dimensions}
$\mathcal{A} \subseteq \{0,\dots,n-1\}$—those axes along which recent movement,
basin detection, or fitness changes suggest the presence of exploitable
structure. Initially all dimensions are active, but the set is progressively
refined as the search proceeds. During each hill-climb step
(Algorithm~\ref{alg:hill-climb-loop}), HC-SHIFT identifies the subset of
dimensions $\mathcal{M}$ that contributed to any observed fitness improvements
(line~12). For every dimension $d\notin\mathcal{M}$, a stagnation counter
$\sigma_d$ is incremented, whereas counters for dimensions in $\mathcal{M}$ are
reset to zero. When $\sigma_d$ exceeds a patience threshold $p$ (default
$p=20$), dimension $d$ is removed from the active set (line~14), indicating
that further exploration along that axis is unlikely to yield progress.

This mechanism ensures that only active dimensions participate in neighbor
generation (Algorithm~\ref{alg:generate-neighbors}), basin detection
(Algorithm~\ref{alg:basin-detection}), and restart selection
(Algorithm~\ref{alg:restart-selection}). As a result, the effective
per-iteration cost is reduced from $O(n)$ to $O(k)$ over active dimensions, where
$k = |\mathcal{A}| \ll n$ denotes the number of remaining active dimensions.
Inactive dimensions retain their previously computed compression metadata but
are no longer revisited unless reactivated later in the search.

For each active dimension $d$, the compression manager maintains a slice-based
map
\[
    \mathrm{dim\_compressions}[d] 
        : (\text{fixed coords}) 
        \mapsto 
        \{(s_1,L_1), (s_2,L_2), \dots \},
\]
where ``fixed coords'' represents all coordinates other than $d$. This design
supports heterogeneous and slice-dependent compression geometries, enabling
HC-SHIFT to adapt to highly anisotropic and irregular fitness landscapes.
Unlike tabu-style or AVM-style strategies that treat all coordinates uniformly
or rely solely on historical movement patterns, HC-SHIFT allocates computation
selectively, focusing effort precisely on the axes that reveal actionable
geometric structure.

\paragraph{(e) Search Procedure and Neighbor Selection.}
\mbox{}
\vspace{0.5em}\\
The HC-SHIFT search procedure alternates between hill climbing in the
compressed $Z$-space, basin detection in the original $X$-space, and
restarts from basin boundaries. During a hill-climb phase
(Algorithm~\ref{alg:hill-climb-loop}), neighbor candidates are generated
according to Algorithm~\ref{alg:generate-neighbors}. For each active
dimension $d$, the current coordinate $x_d$ is mapped into compressed form
$z_d = w_d(x_d)$ whenever a compression mapping exists. Single-step moves
$z_d \pm 1$ are then considered, and the resulting points are mapped back
using the inverse transform $w_d^{-1}$:
\[
    x_d' \in 
    \bigl\{\, w_d^{-1}(z_d - 1),\ w_d^{-1}(z_d + 1)\,\bigr\}.
\]
When no compression is associated with dimension $d$, these updates reduce
naturally to $x_d \pm 1$. As a consequence, seemingly small moves in the
compressed space may correspond to long jumps along flat basins in the
original landscape.

In addition to axis-aligned moves, HC-SHIFT constructs diagonal neighbors
by simultaneously perturbing pairs of active dimensions. For any pair
$(d_1,d_2)$, the corresponding compressed coordinates
$(z_{d_1},z_{d_2})$ are perturbed independently by $\pm 1$ and then mapped
back:
\[
    (x_{d_1}', x_{d_2}')
    =
    \bigl(
        w_{d_1}^{-1}(z_{d_1} \pm 1),\ 
        w_{d_2}^{-1}(z_{d_2} \pm 1)
    \bigr).
\]
This mechanism produces $O(|\mathcal{A}|^2)$ diagonal candidates and allows
the search to traverse narrow or curved passages that would be difficult to
escape using only axis-aligned moves.

Whenever hill climbing reaches a point at which none of the generated
neighbors reduces the fitness, basin detection is invoked
(Algorithm~\ref{alg:basin-detection}) to identify flat regions along each
active dimension. For each detected basin $B_d = [\ell_d, r_d]$, restart
candidates are placed at the geometric boundaries $\ell_d - 1$ and
$r_d + 1$ (Algorithm~\ref{alg:restart-selection}), enabling the algorithm
to reposition the search directly outside the compressed plateau rather
than stepping through it incrementally.

Among all evaluated candidates, HC-SHIFT adopts the steepest-descent
strategy and moves toward the neighbor with the smallest fitness value
(lines~8--16 of Algorithm~\ref{alg:hill-climb-loop}). If a strictly better
neighbor exists, the algorithm continues its ascent in the compressed
space; otherwise, basin detection is triggered. Convergence speed is
recorded as the number of hill-climb steps executed in $Z$-space, a metric
that often underestimates the effective distance traveled in the original
$X$-space because SHIFT compresses large plateaus and oscillatory regions
into compact intervals.


\paragraph{(f) Comparison with Tabu Search and Related Work.}
\mbox{}
\vspace{0.5em}\\
HC-SHIFT differs fundamentally from tabu search in both purpose and mechanism. 
Whereas tabu search modifies the search \emph{trajectory} by maintaining a memory 
structure that prevents revisits, HC-SHIFT instead modifies the \emph{search 
space} itself. Through its SHIFT transformation, flat regions are compressed 
into short intervals, enabling the algorithm to traverse large plateaus with 
only a few effective steps in $Z$-space. Tabu search reduces cycling but does 
not alter the underlying geometry of the landscape and does not reduce the 
effective dimensionality of the problem. In contrast, HC-SHIFT incorporates 
dimension-level pruning via stagnation tracking, allowing it to focus 
computation on the axes that meaningfully influence progress.

Relative to prior SBST techniques such as AVM variants, adaptive step-size 
schemes, and gradient-smoothing methods, HC-SHIFT introduces a distinct 
landscape-shaping perspective. It performs explicit, fitness-driven probing of 
flat regions to estimate basin boundaries through bidirectional detection; it 
applies smooth and invertible sigmoid-based warpings rather than discrete 
mutational shifts; and it maintains all compression metadata in the original 
$X$-space to ensure stability across repeated transformations. Additionally, 
HC-SHIFT applies per-dimension compression independently while dynamically 
selecting active dimensions based on stagnation counters, thereby avoiding 
unnecessary compression on irrelevant axes. The neighbor-generation strategy 
in compressed space further leverages diagonal movements, enabling escape from 
narrow transition corridors that commonly trap local-search methods. Finally, 
HC-SHIFT preserves and accumulates compression information across iterations 
for the same target branch—such as basin boundaries, active dimensions, and 
SHIFT parameters—allowing the algorithm to operate on increasingly informed 
landscapes rather than restarting from scratch at each attempt.

Taken together, these mechanisms constitute a distinct form of 
landscape-shaping local search that proves highly effective on the rugged and 
plateau-rich fitness functions frequently encountered in SBST. A detailed 
empirical comparison with HC and GA is provided in 
Section~\ref{sec:experiment}.


\section{Experiment}
\label{sec:experiment}

\subsection{Pilot Study on Synthetic Fitness Landscapes}
\label{sec:pilot-exp}
Prior to the main evaluation, we conduct a pilot study on synthetic fitness landscapes to verify the fundamental behavior of SHIFT.
This study examines how compression accumulates across trials and how this accumulation affects convergence on landscapes of varying ruggedness and dimensionality.
Across all controlled landscapes considered, SHIFT consistently exhibits superior convergence behavior, requiring notably fewer restarts than HC and surpassing GA in all evaluated settings.
SHIFT’s compression mechanism progressively enhances trial effectiveness, enabling reliable escape from rugged or high-dimensional fitness basins.

\paragraph{Synthetic Fitness Functions}
\label{subsubsec:synthetic-fitness}
We evaluate eight synthetic fitness landscapes—Needle, Plateau, Rugged, and Combined—each instantiated in both 1D and 2D. 
These landscapes span a diverse range of structural challenges for assessing algorithmic behavior. Here, we present two representative examples: a 1D needle landscape in Figure~\ref{fig:needle-1d} and a 2D rugged landscape in Figure~\ref{fig:rugged-2d}.






\begin{figure}[H]
\centering
\begin{subfigure}{0.48\linewidth}
    \centering
    \includegraphics[width=\linewidth]{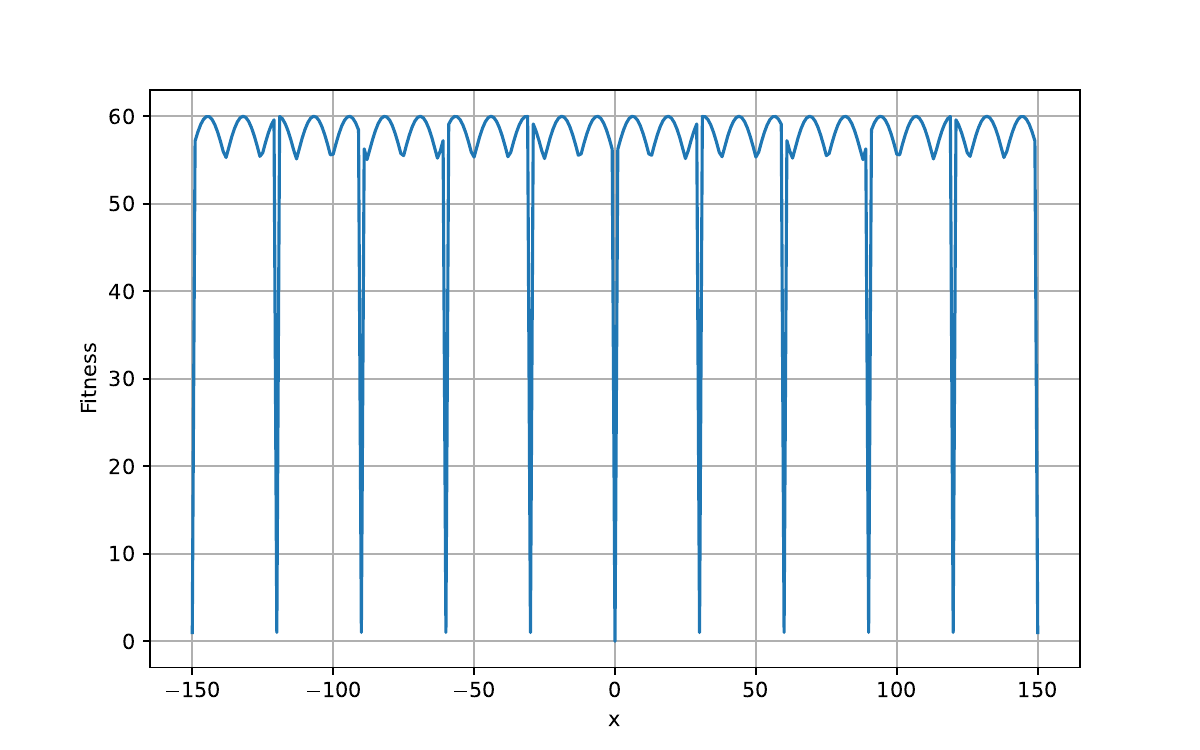}
    \caption{1D Needle landscape.}
    \label{fig:needle-1d}
\end{subfigure}
\hfill
\begin{subfigure}{0.48\linewidth}
    \centering
    \includegraphics[width=\linewidth]{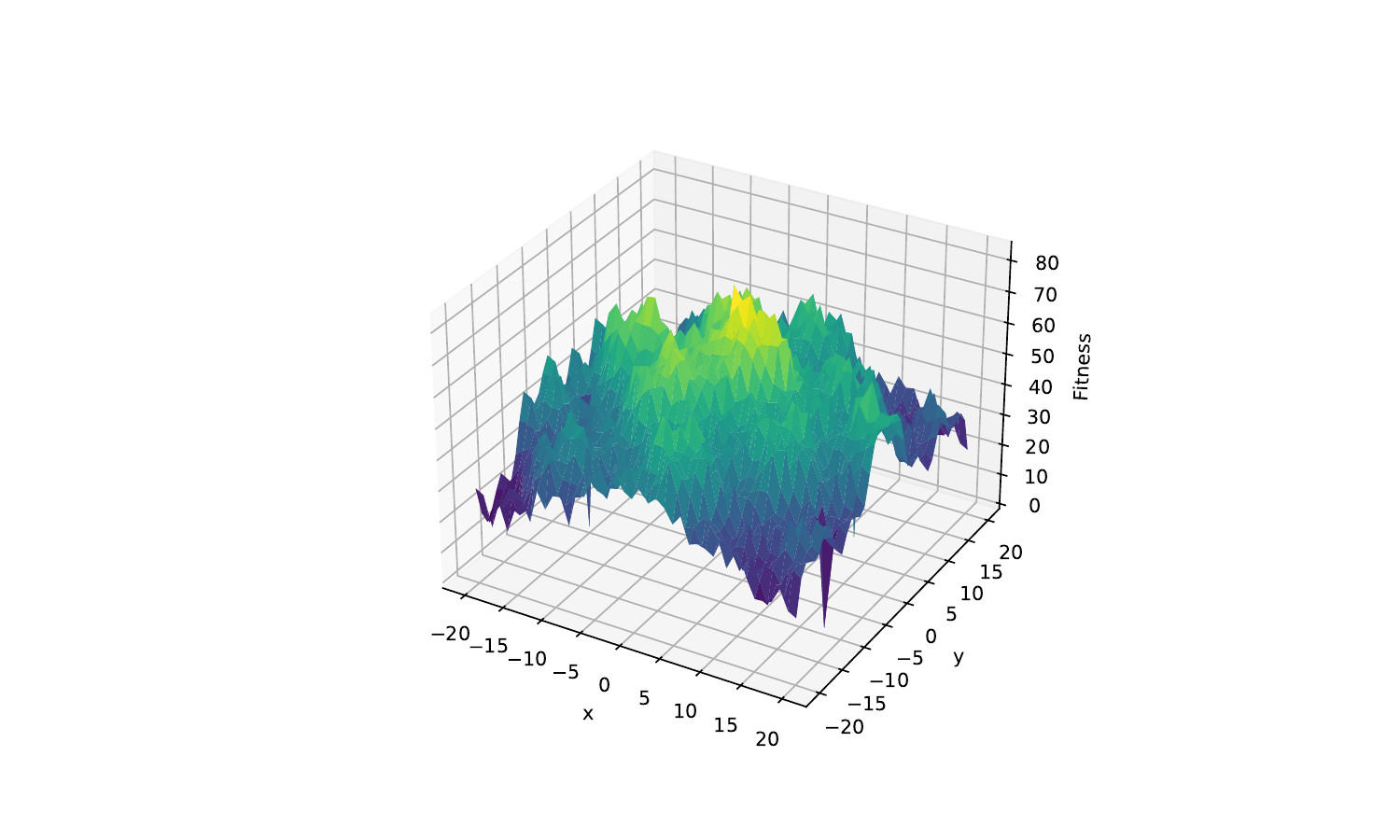}
    \caption{2D Rugged landscape.}
    \label{fig:rugged-2d}
\end{subfigure}

\caption{
Examples of representative synthetic fitness landscapes used in the pilot study.
}
\label{fig:needle-rugged-side-by-side}
\end{figure}

Additional landscapes are included in
Appendix~A for completeness.

\paragraph{Hyperparameter Settings}
\label{subsubsec:hyperparams}
All methods share a common time budget and initialization conditions.  
The full hyperparameter configuration for HC, HC-SHIFT, and GA is summarized in Appendix B.

\subsection{Results}
\subsubsection{Success Rates Across Test Landscapes}
While all algorithms succeed on the 1D landscapes, GA fails to converge within the time budget on the more structurally difficult 2D Needle, Rugged, and Combined landscapes.

\begin{table}[H]
\centering
\setlength{\tabcolsep}{10pt}
\renewcommand{\arraystretch}{1.15}
\begin{tabular}{llccc}
\toprule
\textbf{Dim.} & \textbf{Landscape} & \textbf{GA} & \textbf{HC} & \textbf{HC-SHIFT} \\
\midrule
\multirow{4}{*}{1D}
    & Needle     & p & p & p \\
    & Plateau    & p & p & p \\
    & Rugged     & p & p & p \\
    & Combined   & p & p & p \\
\midrule
\multirow{4}{*}{2D}
    & Needle     & \cellcolor{gray!20} f & p & p \\
    & Plateau    & p & p & p \\
    & Rugged     & \cellcolor{gray!20} f & p & p \\
    & Combined   & \cellcolor{gray!20} f & p & p \\
\bottomrule
\end{tabular}
\caption{
Success (p) or failure (f) of GA, HC, and HC-SHIFT on synthetic landscapes.
Shading indicates cases where GA fails within the time budget.
}
\label{tab:landscape-success-compact}
\end{table}



\subsubsection{Analysis of Trial Counts (Table~\ref{tab:hc-shift-trials-grouped}) and NFE (Table~\ref{tab:hc-hcc-dim-evals})}
HC-SHIFT requires markedly fewer trials and fitness evaluations than standard HC across all synthetic landscapes, with the performance gap widening markedly in 2D settings. While HC exhibits restart-intensive behavior—incurring thousands of trials and tens of thousands of evaluations on rugged or needle-like 2D functions—SHIFT converges rapidly by accumulating compression information that progressively narrows the effective search space. These results demonstrate that SHIFT scales more consistently than HC as dimensionality and landscape complexity increase.

\begin{table}[H]
\centering
\setlength{\tabcolsep}{10pt}
\renewcommand{\arraystretch}{1.2}
\begin{tabular}{llcc}
\toprule
\textbf{Dim.} & \textbf{Landscape} & \textbf{HC} & \textbf{HC-SHIFT} \\
\midrule
\multirow{4}{*}{1D}
    & Needle     & 47    & 1 \\
    & Plateau    & 8     & 1 \\
    & Rugged     & 47    & 1 \\
    & Combined   & 47    & 1 \\
\midrule
\multirow{4}{*}{2D}
    & Needle     & 8198  & 1 \\
    & Plateau    & 8     & 1 \\
    & Rugged     & 11690 & 407 \\
    & Combined   & 1463  & 543 \\
\bottomrule
\end{tabular}
\caption{Number of trials required for convergence on 1D and 2D synthetic landscapes.}
\label{tab:hc-shift-trials-grouped}
\end{table}

\begin{table}[H]
\centering
\setlength{\tabcolsep}{10pt}
\renewcommand{\arraystretch}{1.2}
\begin{tabular}{l l cc}
\toprule
\textbf{Dim.} & \textbf{Landscape} & \textbf{HC} & \textbf{HC-SHIFT} \\
\midrule
\multirow{4}{*}{1D}
    & Needle     & 173   & 5 \\
    & Plateau    & 326   & 93 \\
    & Rugged     & 112   & 23 \\
    & Combined   & 164   & 28 \\
\midrule
\multirow{4}{*}{2D}
    & Needle     & 55119 & 12 \\
    & Plateau    & 372   & 140 \\
    & Rugged     & 38905 & 14160 \\
    & Combined   & 6883  & 5926 \\
\bottomrule
\end{tabular}
\caption{Number of fitness evaluations (NFE) executed until convergence on 1D and 2D synthetic landscapes.}
\label{tab:hc-hcc-dim-evals}
\end{table}

\newpage
\subsubsection{Case Studies on 2D Landscapes}
\label{subsec:case-study}

To better illustrate the qualitative differences between HC and HC-SHIFT, we present
two representative case studies on 2D landscapes: a structurally simple needle-like
surface and a highly rugged multimodal surface.

\paragraph{Case 1: Simple 2D Landscape (Needle 2D).}
As shown in Figure~\ref{fig:case-needle-2d}, SHIFT collapses the large flat or
needle-like domain into a navigable structure, enabling direct convergence to the
global basin in a single trial.  
In contrast, HC exhibits no directional signal in this landscape and
wanders chaotically until repeatedly restarting, ultimately requiring 8,198 trials
to locate the optimum.

\begin{figure}[H]
\centering
\includegraphics[width=0.48\linewidth]{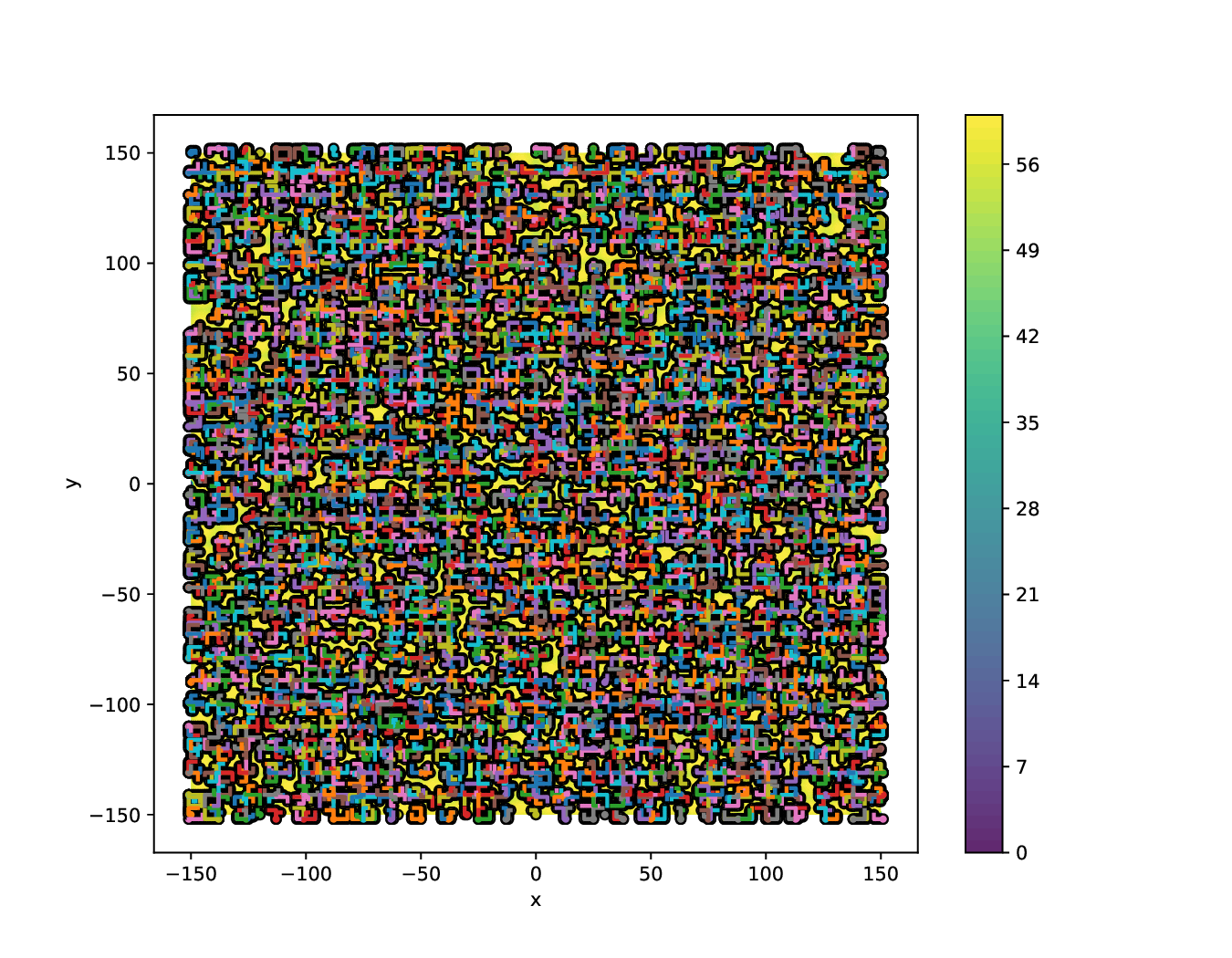}
\includegraphics[width=0.48\linewidth]{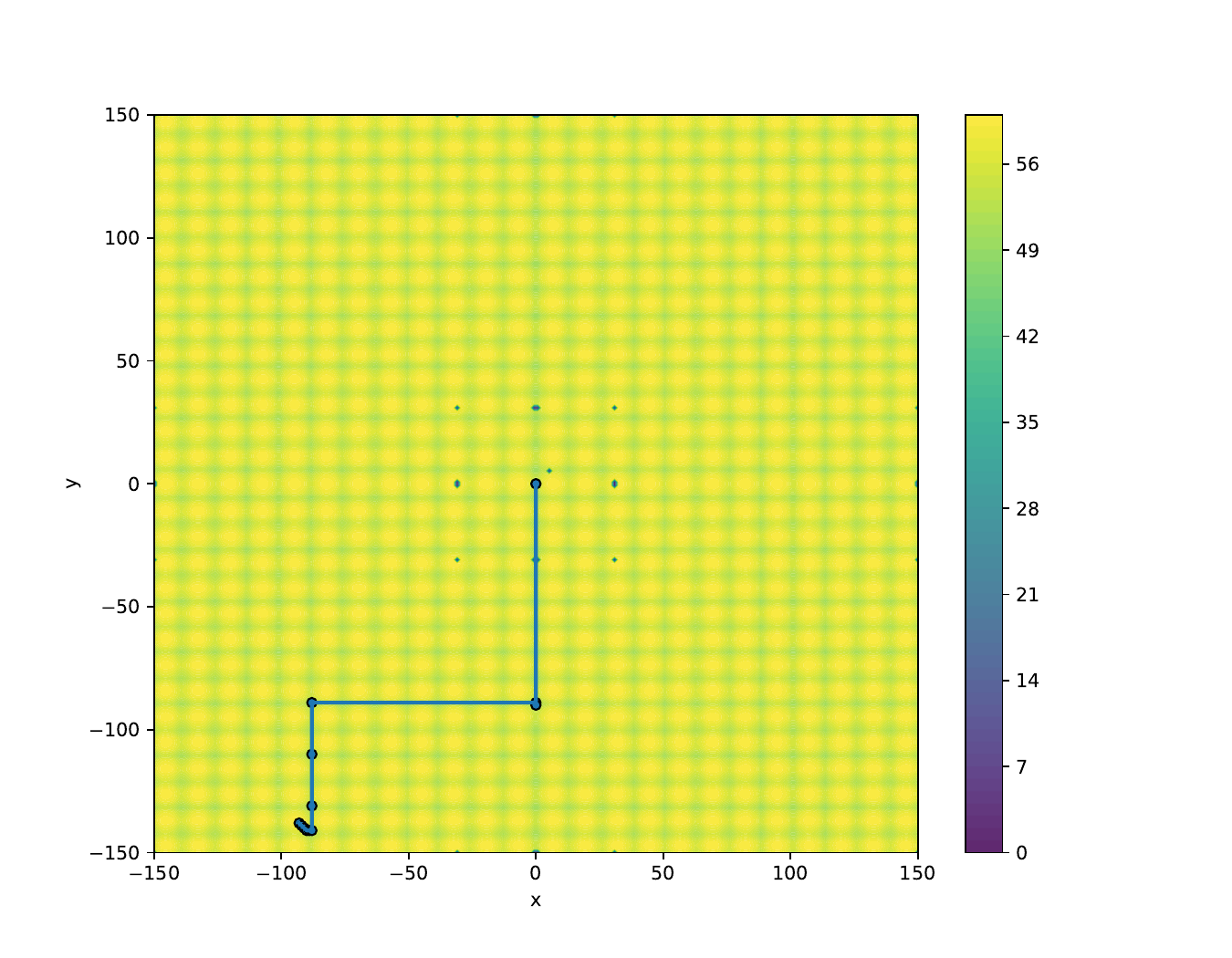}
\caption{Search trajectories of HC (left) and HC-SHIFT (right) overlaid on the contour map of the Needle 2D landscape.}
\label{fig:case-needle-2d}
\end{figure}

\paragraph{Case 2: Rugged High-Dimensional Landscape (Rugged 2D).}
On the highly irregular Rugged 2D landscape (Figure~\ref{fig:case-rugged-2d}), HC
becomes trapped almost immediately, requiring an extremely large number of restarts
(11{,}690 trials) due to its inability to escape dense local minima.
SHIFT, however, progressively compresses the landscape, reducing the effective
search space at each restart and ultimately converging after 407 trials.
The characteristic grid-like trajectory arises because SHIFT compresses each
dimension independently, causing the algorithm to alternate axis-aligned moves
as the underlying landscape is incrementally flattened.

\begin{figure}[H]
\centering
\includegraphics[width=0.48\linewidth]{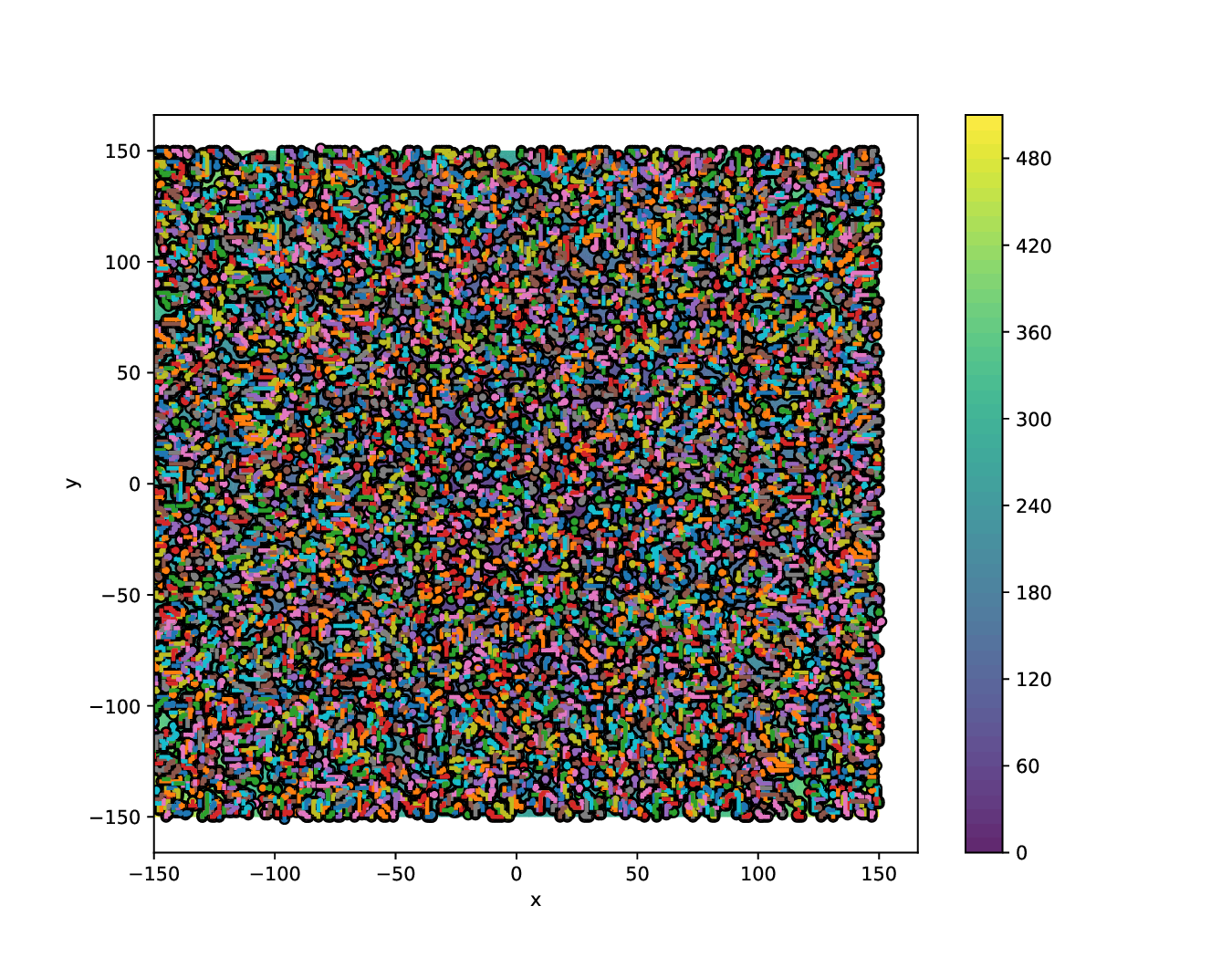}
\includegraphics[width=0.48\linewidth]{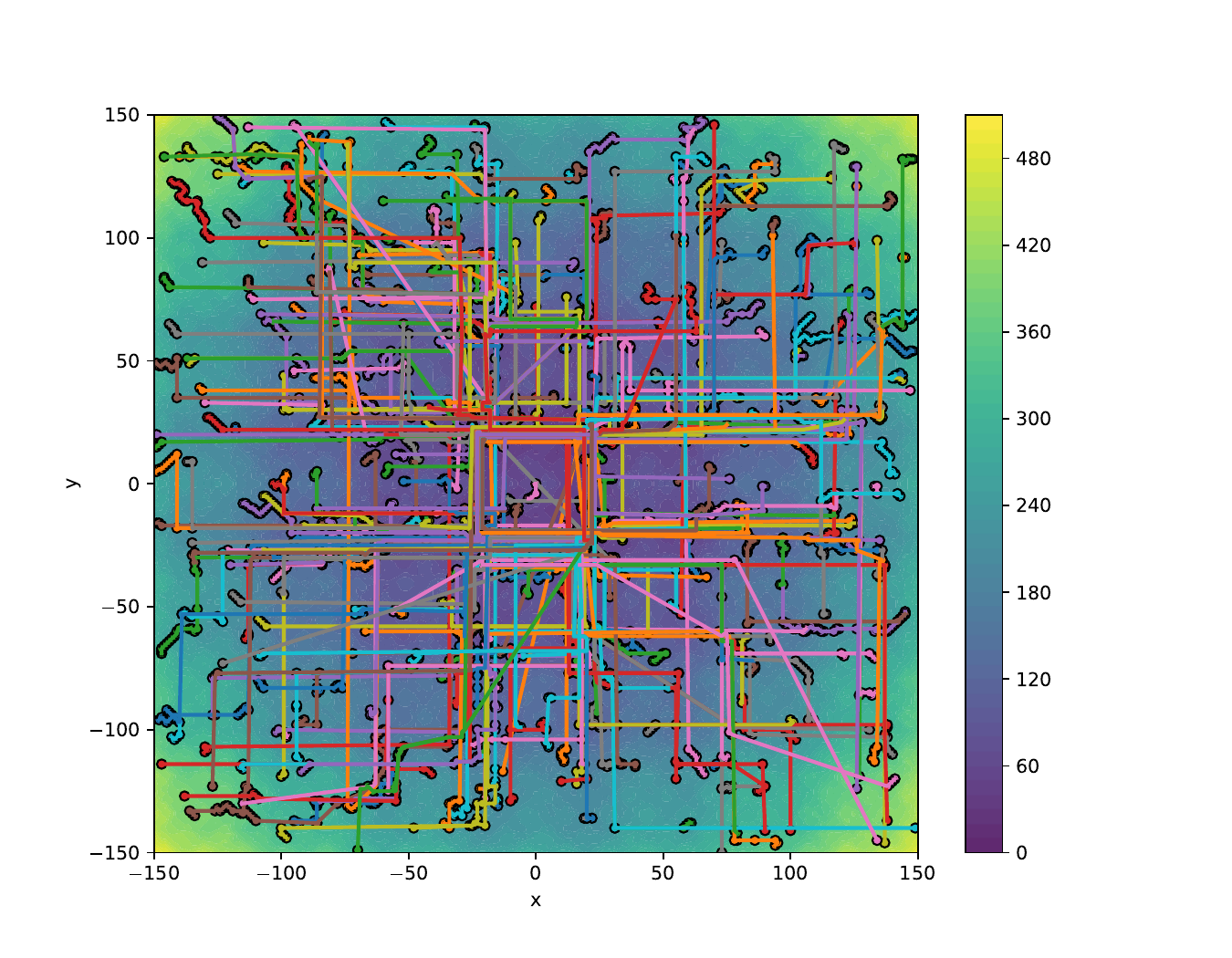}
\caption{Search trajectories of HC (left) and HC-SHIFT (right) on the contour map of the Rugged 2D landscape.}
\label{fig:case-rugged-2d}
\end{figure}

Overall, the pilot study confirms that SHIFT’s compression-driven adaptation yields markedly faster and more stable convergence than HC and GA, motivating its use in the following real-world experiments.

\newpage
\subsection{Benchmark Programs and Branch Extraction}



\paragraph{Benchmark Programs}
To evaluate the efficacy of the proposed sigmoid-based compression, we constructed and curated a benchmark suite of 38 Python programs comprising 343 branching points. Semantically unreachable branches—such as the \texttt{False} direction of a \texttt{range}-based \texttt{for}-loop—are excluded from the count. These programs were deliberately designed or selected to exhibit complex control-flow dependencies and structurally challenging fitness landscapes, including rugged terrains, extended plateaus, needle-in-a-haystack configurations, and combinations thereof—landscape types that are well-known to be problematic for conventional SBST algorithms. The inherent difficulty of these benchmarks provides a rigorous testbed for assessing the performance gains afforded by our approach.

The benchmark programs are broadly classified into the following categories based on the primary challenge each presents to a search algorithm; the distribution is summarized in Table~\ref{tab:benchmarks}. The \textit{Other} category encompasses programs with general control flow that do not fall into any of the specialized categories.

\begin{table}[H]
\centering
\begin{tabular}{lccc}
\toprule
\textbf{Category} & \textbf{\# Programs} & \textbf{\# Branch} \\
\midrule
Plateau & 4 & 40 \\
Rugged & 3 & 22 \\
Needle-in-a-haystack & 3 & 13 \\
Mixed and Complex & 3 & 42 \\
Other & 25 & 226 \\
\bottomrule
\end{tabular}
\caption{Benchmark categories and number of programs/branches.}
\label{tab:benchmarks}
\end{table}

Each category has its own special landscape feature, making it difficult to generate test cases automatically.
\begin{itemize}
    \item \textbf{Plateau Landscapes}
    
    Plateaus are regions in the search space where the fitness function yields a constant value, providing no gradient for a local search algorithm to follow. 


    \item \textbf{Rugged Landscapes}
    
    A rugged landscape is one where small changes to the input can cause large, unpredictable changes in the fitness value. This noise can mislead the search algorithm.


    \item \textbf{Needle-in-a-Haystack Landscapes}
    
    These problems feature a fitness landscape that is almost entirely flat, with a single, extremely narrow basin of attraction leading to the solution. The algorithm receives little to no guidance until it is very close to the target.


    \item \textbf{Mixed and Complex Landscapes}
    
    Several benchmarks combine multiple characteristics to create a composite challenge.

\end{itemize}

\paragraph{Branch Extraction}
For each benchmark program, we performed static analysis to identify all conditional branches connecting other program components. Each conditional statement (\texttt{if}, \texttt{while}, \texttt{for} and \texttt{match}) in the source code is identified as a branching point in the control-flow graph and stored to generate a test case for each branch. The goal of the test generation process is to generate inputs that cover these branches. The branch distance, a core component of the fitness function in SBST, is calculated at each conditional. For our experiment, we instrumented the Python code to report the outcome of each predicate, enabling the search algorithm to compute the fitness of a given input based on how closely it satisfied the conditions for a target branch. This instrumentation provides the necessary feedback for the search algorithm to navigate the program's control flow.

\subsection{Experimental Protocol}
\label{sec:protocol}

We evaluate three search algorithms—HC, HC-SHIFT, and GA—under a unified,
time-budgeted per-branch evaluation framework.  
This ensures fairness across algorithms with different iteration costs and
guarantees reproducibility through consistent initialization and random
seed management.

\subsubsection{Per-Branch Time Budget}
\label{sec:timebudget}

Each conditional branch is tested independently under a strict wall-clock
limit of $T = 20$ seconds.  
A fixed global random seed (42) is reapplied at the beginning of every 
branch evaluation.  
The timer starts immediately before the first call to the branch-level
fitness function.

In HC and HC-SHIFT, multiple restart trials are executed sequentially
until the branch-level time budget expires.  
Each trial begins from a newly sampled initial point and continues until
a local optimum or an internal stopping condition is reached.

In GA, the population is initialized once and evolved continuously.
Before evaluating any individual, a time check is performed; if the time
budget is exceeded, GA stops immediately and returns the best individual
observed so far.  
This strict time-guarding ensures a fair comparison between HC, 
HC-SHIFT, and GA.

\subsubsection{Initialization Strategies}
\label{sec:init}

All algorithms share the same input initialization mechanism, applied
per trial (HC, HC-SHIFT) or per individual (GA).  
Two initialization modes are supported:

\paragraph{Random Initialization.}
Each input variable is sampled uniformly from the automatically inferred
domain $[L,H]$ derived from program constants.

\paragraph{Biased Initialization (Default).}
We extract variable-specific and global constants from the program’s AST
and sample values using a mixture model:
\[
x_0 \sim 
\begin{cases}
\mathsf{Uniform}(L,H), & \text{with prob.\ } 0.2,\\[4pt]
\mathsf{Gaussian}(c,\sigma^2), 
    & \text{with prob.\ } 0.8,\ c\in\text{Constants},
\end{cases}
\]
with $\sigma$ set to $1\%$ of the domain range.

\subsubsection{Metrics}
\label{sec:metrics}

For every branch and target outcome, we record one CSV row per algorithm.
HC and HC-SHIFT use the following fields:

\begin{center}
\texttt{function, lineno, outcome, convergence\_speed, nfe,}\\
\texttt{best\_fitness, best\_solution, success, num\_trials,}\\
\texttt{total\_time, time\_to\_solution}
\end{center}

GA uses the same fields plus:

\begin{center}
\texttt{generations}
\end{center}

The fields have the following meaning:

\begin{itemize}[nosep]
    \item \textbf{function}: name of the function containing the target branch.
    \item \textbf{lineno}: branch identifier (bid) assigned during
          instrumentation (Section~\ref{sec:instrumentation}).
    \item \textbf{outcome}: target direction (\texttt{True} or \texttt{False}).
    \item \textbf{best\_fitness}: minimum value of $(AL + \mathrm{nBD})$.
    \item \textbf{best\_solution}: input that achieved the best fitness.
    \item \textbf{success}: \texttt{True} iff \texttt{best\_fitness} is zero.
    \item \textbf{convergence\_speed}:  
          for HC, the number of accepted moves in $X$-space;  
          for HC-SHIFT, the number of accepted moves in $Z$-space;  
          for GA, the generation index where the best individual first appears.
    \item \textbf{nfe}: number of fitness evaluations.
    \item \textbf{num\_trials}:  
          for HC and HC-SHIFT, the number of restart trials attempted;  
          for GA, the number of evaluated individuals.
    \item \textbf{generations} (GA only): number of completed generations.
    \item \textbf{total\_time}: wall-clock time used (capped at $T$).
    \item \textbf{time\_to\_solution}: time until a successful input was found.
\end{itemize}

\newpage
\subsubsection{Hyperparameter Settings}
\label{sec:hyperparams}

Table~\ref{tab:hyperparams} lists the hyperparameters used across all
algorithms.

\begin{table}[H]
\centering
\begin{tabular}{lll}
\toprule
\textbf{Component} & \textbf{Parameter} & \textbf{Value} \\
\midrule
\multirow{4}{*}{Common} 
    & Time budget per branch $T$ & $20\ \mathrm{s}$ \\
    & Random seed & $42$ \\
    & Input domain $[L,H]$ & auto-detected (AST constants) \\
    & Initialization mode & biased \\
\midrule
\multirow{2}{*}{HC}
    & Max steps per trial & $2000$ \\
    & Move set & $\pm 1$ per dimension \\
\midrule
\multirow{3}{*}{HC-SHIFT}
    & Max iterations per dimension & $10$ \\
    & Basin max search per dimension & $1000$ \\
    & Active-dimension updates & enabled \\
\midrule
\multirow{4}{*}{GA}
    & Population size $N$ & $10000$ \\
    & Tournament size $k$ & $3$ \\
    & Elite ratio & $0.1$ \\
    & Mutation steps & $\{-3,-2,-1,1,2,3\}$ \\
\bottomrule
\end{tabular}
\caption{Hyperparameter settings for HC, HC-SHIFT, and GA.}
\label{tab:hyperparams}
\end{table}

\paragraph{Parameter Selection.}
These hyperparameters were determined through preliminary tuning experiments
on representative benchmark subsets to ensure fair comparison:

\begin{itemize}[nosep]
    \item \textbf{GA population size:} Increased from 1,000 to 10,000 to
          provide sufficient diversity for exploring complex fitness
          landscapes within the 20-second time budget.
    
    \item \textbf{HC-SHIFT basin detection:} Tested configurations
          $(10, 10000)$ and $(100, 100)$ for (max iterations per dimension,
          basin max search per dimension). The selected $(10, 1000)$
          configuration achieved the best balance between exploration depth
          and computational efficiency.
    
    \item \textbf{HC parameters:} Retained standard settings ($\pm 1$ moves,
          2000 max steps) as preliminary experiments confirmed adequate
          performance as a baseline method.
\end{itemize}

These settings serve as the default configuration across all experiments.
Individual experiments (Section~\ref{sec:experiment}) may vary specific
parameters to investigate particular research questions; such variations
are explicitly noted in the corresponding experiment descriptions.

\subsection{Experiment 1: Are Random Restarts Sufficient?}
\label{sec:exp1}

A natural question in search-based testing is whether simply increasing
the number of random restart trials is sufficient to achieve high coverage
within a fixed time budget.  
This experiment addresses that question by comparing the three algorithms
introduced in Section~\ref{sec:approach}—HC, GA, and HC-SHIFT—across all
benchmark programs in the \texttt{./benchmark} directory, totaling 38 source
files and 3,800 branch outcomes.

Each branch is tested independently under a strict per-branch time budget
of 20 seconds (Section~\ref{sec:timebudget}).  
Coverage is computed as
\[
\mathrm{Coverage}
= \frac{\text{\# branches covered}}{3800} \times 100\%.
\]
Both initialization modes—biased and pure random (Section~\ref{sec:init})—were
evaluated for all algorithms.

Table~\ref{tab:exp1-coverage} reports the aggregated results.

\begin{table}[H]
\centering
\begin{tabular}{lccc}
\toprule
\textbf{Algorithm} & \textbf{Init.} & \textbf{Coverage} & \textbf{Runtime} \\
\midrule
HC-SHIFT 
  & biased 
  & 3642 / 3800 (95.84\%) 
  & 3m 39s \\
HC-SHIFT
  & random 
  & 3569 / 3800 (93.92\%) 
  & 4m 13s \\
HC
  & biased 
  & 3586 / 3800 (94.42\%) 
  & 4m 32s \\
HC
  & random 
  & 3519 / 3800 (92.60\%) 
  & 5m 03s \\
GA
  & biased 
  & 3519 / 3800 (92.60\%) 
  & 4m 55s \\
GA
  & random 
  & 3491 / 3800 (91.87\%) 
  & 4m 54s \\
\bottomrule
\end{tabular}
\caption{Coverage comparison under a 20-second per-branch time budget.}
\label{tab:exp1-coverage}
\end{table}

\paragraph{Key Findings.}
HC-SHIFT achieves the highest coverage across all configurations.  
With biased initialization it reaches \textbf{95.84\%}, outperforming both
HC and GA.  
Even under pure random initialization, it maintains strong performance
(\textbf{93.92\%}), demonstrating robustness to initialization quality.

\paragraph{Trial Efficiency: The Answer to Random Restarts.}
The critical distinction emerges when examining the number of trials required
to achieve coverage.  
Table~\ref{tab:benchmark_challenging} presents selected challenging benchmarks
(full results in Appendix~\ref{appendix:benchmark_res}) where HC-SHIFT
typically converges in \textbf{2--8 trials on average}, whereas HC and GA
require \textbf{hundreds of thousands of trials} and still fail to achieve
full coverage on plateau-dominated functions.

This efficiency gap demonstrates that \emph{random restarts alone are
insufficient}.  
Both HC and GA repeatedly sample similar regions without learning from
prior failures, exhausting their time budget through redundant exploration.

\begin{table*}[!htbp]
\hspace*{0cm}   
{
\centering
\resizebox{\textwidth}{!}{%
\small
\begin{tabular}{l|ccc|ccc|ccc}
\toprule
\textbf{Sub Benchmark}
& \multicolumn{3}{c|}{\textbf{HC-SHIFT}}
& \multicolumn{3}{c|}{\textbf{HC}}
& \multicolumn{3}{c}{\textbf{GA}} \\
& Avg.\# & Avg.t (s) & Cov. & Avg.\# & Avg.t (s) & Cov.  & Avg.\# & Avg.t (s) & Cov. \\
\midrule
count\_divisor\_2 & 2.58 & 4.145 & 100\% & 8312.83 & 6.674 &  88\% & 28765.00 & 6.829 &  88\%\\
mixed\_case & 1.00 & 1.672 & 100\% & 13.67 & 5.782 &  96\% & 141325.50 & 8.451 &  78\%\\
plateau2 & 3.75 & 0.572 & 100\% & 106216.25 & 5.005 &  88\% &  233190.88 & 5.056 &  88\%\\
plateau\_case & 8.50 & 4.167 & 100\% & 86845.62 & 10.057 &   83\% & 267704.38 & 12.584 &  57\%\\
\bottomrule
\end{tabular}
}%
\caption{Selected challenging benchmarks showing HC-SHIFT's superior 
trial efficiency (subset of full results in Appendix~\ref{appendix:benchmark_res}).}
\label{tab:benchmark_challenging}
}
\end{table*}

\paragraph{Mechanisms Enabling HC-SHIFT's Efficiency.}
HC-SHIFT's superior performance arises from the interaction of several 
complementary mechanisms. First, the basin detection phase 
(Section~\ref{sec:compression}(b)) explicitly identifies the geometric 
structure of flat or weakly varying regions, in contrast to HC and GA, 
which continue probing blindly. Once such basins are detected, the 
invertible SHIFT compression (Section~\ref{sec:compression}(c)) contracts 
these plateaus into short intervals, enabling the algorithm to traverse 
areas that normally induce stagnation with only a few steps in the 
compressed space. 

The restart policy further amplifies this effect: instead of restarting 
from arbitrary points, HC-SHIFT deliberately evaluates basin-boundary 
locations (Section~\ref{sec:compression}(f), 
Algorithm~\ref{alg:restart-selection}), thereby exploring meaningful 
escape routes from local minima. Finally, compression metadata is preserved 
across trials (Section~\ref{sec:compression}(d--e)), allowing information 
about previously explored regions to accumulate rather than being discarded. 
Together, these mechanisms convert naive random restarts into 
\emph{informed restarts}, leading to higher coverage and markedly
reduced trial counts under the same time budget.

\paragraph{Impact of Initialization Strategies.}
Initialization strategy also affects performance, though to varying degrees 
across algorithms. Biased initialization (Section~\ref{sec:init}) benefits 
all methods, particularly GA, whose early generations depend heavily on the 
quality of initial samples. HC-SHIFT, however, shows the least sensitivity 
to initialization: even when starting from disadvantageous positions, the 
combination of basin detection and compression rapidly reshapes the 
landscape, steering the search toward productive areas. This robustness 
underscores that HC-SHIFT's improvements stem primarily from its landscape 
transformation rather than from favorable initialization.

\newpage
\subsection{Performance on Complex Landscapes: Plateau and Rugged Cases}

A central contribution of HC-SHIFT is its ability to reshape the fitness landscape in order to escape local optima and traverse plateaus efficiently. We evaluate this capability using parameterized benchmark functions designed to represent two archetypal landscape challenges common in software testing: stubborn local optima (Rugged) and large flat regions (Plateau).

\subsubsection{Benchmark Specification}
To rigorously evaluate the algorithms, we designed parameterizable benchmark functions that allow precise control over landscape difficulty—specifically, the length of a plateau or the frequency of local optima.

\paragraph{Plateau Benchmark: \texttt{schedule\_cycle}.}
This function simulates a needle-in-a-haystack scenario common in scheduling logic, where a specific condition must be satisfied within a large, flat search space.
\begin{lstlisting}[language=Python]
def schedule_cycle_N(timestamp: int) -> bool:
    # A = cycle length (e.g., 7 days), B = granularity (e.g., 3600s)
    if (timestamp // 86400) % A == (A - 1):
        if (timestamp % 86400) // B == C:
            return True 
    return False
\end{lstlisting}
By adjusting parameters $A$ (cycle length in days) and $B$ (granularity), we control the plateau length. A larger cycle or finer granularity produces a vast region with zero gradient information ($\mathrm{fitness}=0$ or constant), rendering standard gradient-based guidance ineffective.

\paragraph{Rugged Benchmark: \texttt{rugged\_period}.}
This function introduces periodic noise to a simple linear distance function, creating multiple local optima (peaks and valleys) that trap gradient-based methods.
\begin{lstlisting}[language=Python]
import math
TARGET = 31415

def rugged_period_N(x: int):
    period = N  # Controls ruggedness frequency
    d = abs(x - TARGET)
    # Add sinusoidal noise to create local optima
    noise = int(10 * abs(math.sin((x - TARGET) / period)))
    g = d + noise 
    if g == 0:
        return 0 # Success
    else:
        return g # Fitness distance
\end{lstlisting}

The \texttt{period} parameter governs the density of local optima: a smaller value produces high-frequency noise with many peaks, while a larger value yields fewer, broader peaks.

\subsubsection{Performance Analysis}

\begin{figure}[H]
    \centering
    \includegraphics[width=0.75\textwidth]{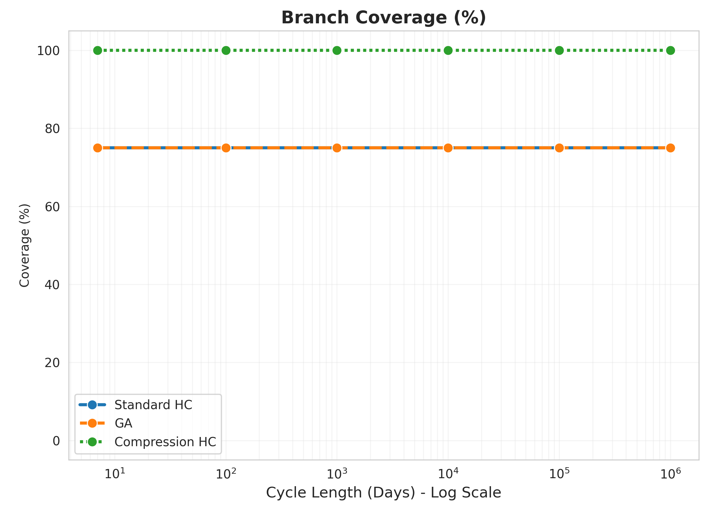}
    \caption{Performance comparison on plateau case}
    \label{fig:plateau}
\end{figure}

\paragraph{Plateau Analysis.}
As shown in Figure~\ref{fig:plateau}, standard HC fails on benchmarks with long plateau landscapes owing to the absence of gradient guidance. HC-SHIFT, by contrast, maintains full coverage regardless of plateau length. This robustness stems from its compression logic, which effectively ``folds'' flat regions—identified by equal fitness values—treating the entire plateau as a compact interval in the compressed space.

\begin{figure}[H]
    \centering
    \includegraphics[width=0.75\textwidth]{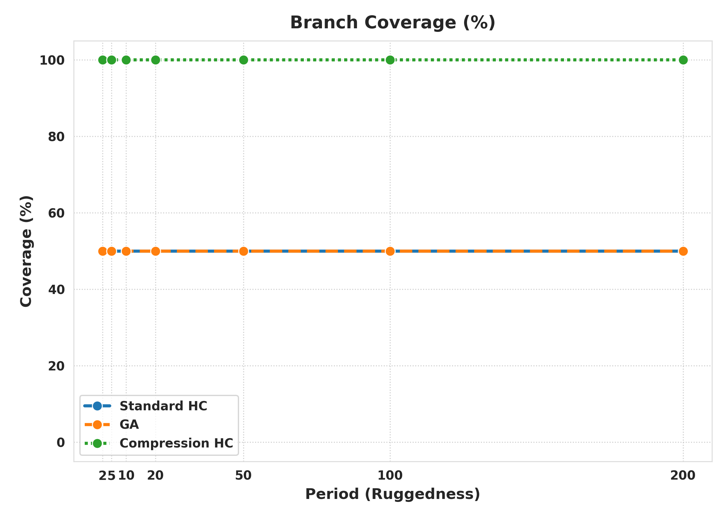}
    \caption{Performance comparison on rugged case}
    \label{fig:rugged}
\end{figure}

\paragraph{Rugged Analysis.}
Figure~\ref{fig:rugged} illustrates the effect of landscape ruggedness, which is governed by the \texttt{period} parameter.

By varying this parameter, we control the number of local optima (peaks) within a fixed landscape extent, allowing us to measure each algorithm’s ability to escape local traps.

Standard HC and GA become trapped at the first local optimum they encounter. HC-SHIFT, by contrast, succeeds in these scenarios by identifying the basin of attraction of the local optimum, applying sigmoid compression to contract it, and stepping past the peak to resume descent toward the target.

The compression phase does incur additional computational cost, and failures were observed when the 20-second time budget proved insufficient. Nevertheless, given adequate time, HC-SHIFT achieves 100\% coverage across all tested configurations, as shown in Figure~\ref{fig:rugged}.

\subsection{Experiment 3: Necessity of `Active' Dimension Strategy}

While compression is effective, it incurs computational cost. To prevent time-budget exhaustion caused by compressing irrelevant dimensions—as observed in the high-frequency rugged case—we employ the \texttt{active\_dim} strategy, which selectively restricts compression to dimensions that meaningfully influence fitness.

\subsubsection{Benchmark Specification}
To assess the robustness of HC-SHIFT under increasing input dimensionality, we selected a benchmark that all algorithms solve easily and augmented it with irrelevant input variables ranging from 1 to 200.

\begin{lstlisting}[language=Python]
def active_easy_100(x0: int, x1: int, x2: int, ...,
                    x98: int, x99: int):
    # content of arbitrary1.py
    return x
\end{lstlisting}

The growing number of irrelevant dimensions introduces spurious compression candidates, potentially exhausting the time budget on dimensions that do not contribute to fitness improvement.

\subsubsection{Performance Analysis}

\begin{figure}[H]
    \centering
    \includegraphics[width=0.8\textwidth]{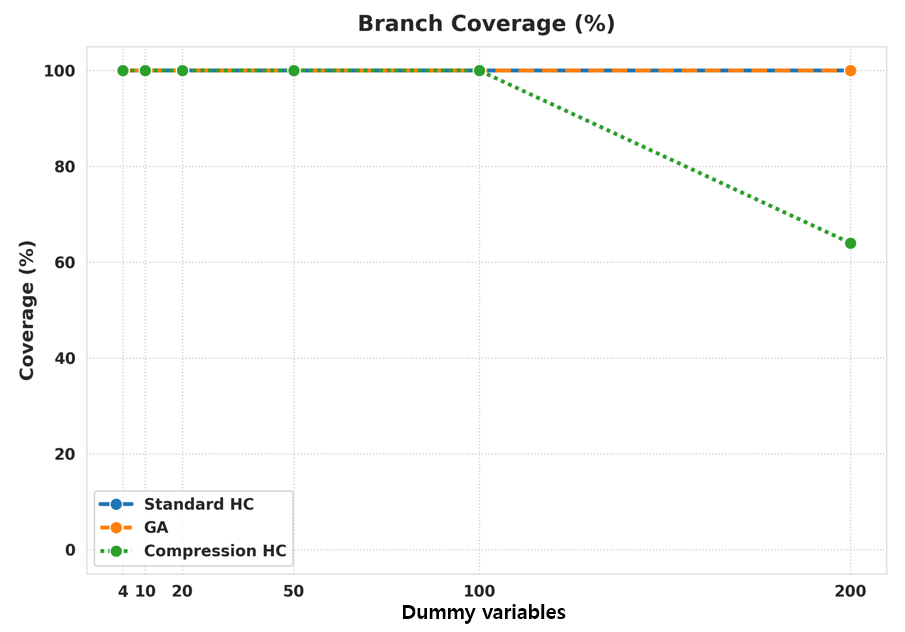}
    \caption{Effectiveness of Active Dimension logic against increasing irrelevant variables.}
    \label{fig:active_dim}
\end{figure}

As shown in Figure~\ref{fig:active_dim}, HC-SHIFT maintains 100\% coverage on the vast majority of test configurations. The \texttt{active\_dim} logic successfully filters out non-contributing dimensions, preventing unnecessary compression attempts and concentrating the time budget on the relevant subspace. A single failure is observed at 200 irrelevant variables, where the time budget is insufficient for the dimension-deactivation phase to complete. Nevertheless, HC-SHIFT remains consistently more robust than standard HC and GA, particularly on plateau and rugged landscapes.

\subsection{Experiment 4: Seed Dependency Test}
\label{sec:exp4}

To evaluate the robustness of each method with respect to initialization randomness, 
we conduct a seed dependency test. 
Each method is executed with seeds ranging from 41 to 50, and for every run we measure 
the total execution time and final coverage.  
We report the mean, standard deviation, and coefficient of variation (CV = std/mean) 
to quantify each method's sensitivity to seed variation.

\begin{table}[H]
\centering
\resizebox{\textwidth}{!}{%
\setlength{\tabcolsep}{8pt}
\renewcommand{\arraystretch}{1.2}
\begin{tabular}{lcccccc}
\toprule
\textbf{Method} 
& \textbf{Time mean} 
& \textbf{Time std} 
& \textbf{Time CV} 
& \textbf{Cov mean} 
& \textbf{Cov std} 
& \textbf{Cov CV} \\
\midrule
SHIFT (biased)  & 228.50 & 20.97 & 0.092 & 95.78 & 0.18 & 0.002 \\
HC (biased)     & 272.60 & 3.67  & 0.013 & 94.37 & 0.00 & 0.000 \\
GA (biased)     & 292.20 & 5.23  & 0.018 & 92.67 & 0.13 & 0.001 \\
\midrule
SHIFT (random)  & 267.70 & 22.27 & 0.083 & 94.03 & 0.41 & 0.004 \\
HC (random)     & 305.50 & 1.86  & 0.006 & 92.89 & 0.23 & 0.002 \\
GA (random)     & 303.80 & 8.74  & 0.029 & 91.66 & 0.42 & 0.005 \\
\bottomrule
\end{tabular}
}%
\caption{Seed dependency results for biased and random initialization.  
CV denotes the coefficient of variation.}
\label{tab:seed-dependency}
\end{table}

\paragraph{Results and Analysis.}

Across all configurations, SHIFT achieves the fastest average execution time and the highest coverage, 
demonstrating clear performance advantages over HC and GA.  
Although SHIFT exhibits greater variance than HC, its substantial gains in both coverage and convergence speed 
far outweigh this variability.

HC shows extremely low variance---a consequence of its deterministic behavior---but it consistently underperforms 
SHIFT in both efficiency and effectiveness.  
GA delivers the weakest performance, with lower coverage and slower execution than both HC and SHIFT, and its higher 
variance (especially under random initialization) indicates strong seed dependency.

These results indicate that SHIFT is the only method that simultaneously achieves high coverage, fast convergence, and robustness across seeds, whereas HC trades performance for stability and GA exhibits neither stability nor competitive performance.

\section{Limitations and Future Works}

Although the proposed sigmoid-based compression framework demonstrates clear benefits on rugged and plateau-heavy fitness landscapes, several limitations remain. Addressing these limitations will broaden the applicability of our method and facilitate deeper integration with a wider range of search and optimization problems.

\subsection*{Limited Input Domain and Program Support}

The current implementation is restricted to Python programs that accept \textit{integer} inputs. While this constraint simplified the design of the compression module, it considerably narrows the scope of programs amenable to analysis. Many real-world testing targets involve floating-point inputs, mixed numeric types, or structured data such as lists and records. Extending the compression scheme to operate reliably on continuous domains—in particular, floating-point inputs—is a direct and necessary next step.

Because floating-point domains introduce smoother yet potentially more deceptive search surfaces, the compression function may require additional refinements, such as adaptive scaling or stability constraints, to ensure consistent behavior. Future work could generalize the model to richer input types and more realistic program structures, thereby increasing the external validity of our approach.

\subsection*{Adaptive Maximum Compression Length}

A key limitation of the current SHIFT framework is its reliance on a statically chosen maximum compression length. Because this cap is fixed beforehand, the model cannot fully exploit plateaus or wide basins that exceed the preset threshold, nor can it scale down compression when the terrain becomes narrow or irregular. This mismatch keeps SHIFT from fully adapting to the true geometry of each search region, leaving potential performance gains unused and forcing manual hyperparameter tuning.

A promising direction for addressing this issue is to make compression length adaptive through lightweight symbolic analysis performed before hill climbing. By examining guard conditions and extracting approximate constraint ranges for each variable, the model can predict where broad plateaus or complex basin structures are likely to appear. These symbolic estimates provide an informed prior on the landscape width, allowing the advanced SHIFT model to dynamically adjust its maximum compression length, expanding it for wide, flat regions and tightening it around constrained or nonlinear areas. This integration would make SHIFT genuinely landscape-aware and more robust across diverse programs without requiring hand-engineered compression parameters.

\subsection*{Applicability Beyond SBST: Hyperparameter Optimization}

Our empirical analysis shows that compression significantly improves search in landscapes dominated by abrupt cliffs and large plateaus. These properties are also typical in hyperparameter optimization for machine learning or deep learning models, where many influential parameters, such as layer counts, batch sizes, or tree depths, are discrete and yield highly discontinuous fitness landscapes.

Given this similarity, our compression framework can naturally extend to \emph{ML hyperparameter optimization}. Traditional hill climbing, greedy tuning, and random search methods often stall in local optima due to the lack of informative local gradients. By compressing the hyperparameter space and reducing the prevalence of unproductive local optima, our method may provide a lightweight yet effective enhancement to existing search loops. Integrating the compression transformation into general hyperparameter tuning pipelines represents a promising direction for future experimentation.

Overall, these limitations highlight meaningful paths forward. By expanding input support, applying the model to other optimization domains, or even integrating it with a wider class of metaheuristic algorithms, we aim to develop a general-purpose landscape compression framework that enhances search efficiency across diverse applications.

\section{Conclusion}

This work addresses a practical weakness of Search-Based Software Testing: traditional hill climbing stalls on rugged landscapes and extended plateaus because the fitness function provides little or no informative gradient. Our core motivation is straightforward—reshape the landscape rather than repeatedly failing on it. To that end, we introduce SHIFT, an invertible sigmoid-based compression framework that contracts flat or weakly varying regions, exposes actionable gradients, and enables even simple hill climbing to escape stagnation.

The core idea behind SHIFT is equally simple: detect basins, compress them, and restart outside. The algorithm first runs hill climbing in the original space, performs bidirectional basin detection to locate contiguous plateaus or near-flat regions, applies a smooth and fully invertible warping to shrink these basins in a separate coordinate space, and then continues the search using the compressed geometry. Along the way, SHIFT tracks active dimensions to avoid wasting time on irrelevant variables and accumulates compression metadata across restarts, effectively learning the landscape as it progresses.

The experiments validate the design across synthetic, constructed, and real benchmark programs. On synthetic landscapes, SHIFT consistently required fewer trials and evaluations than standard hill climbing, with the performance gap widening in higher dimensions. It handled needle-like, plateau-dominated, and rugged multimodal terrains with high reliability, while GA and HC both deteriorated or failed outright in the more extreme cases. On realistic SBST targets, SHIFT achieved the highest overall coverage under identical time budgets, showing strong gains on programs with structurally difficult branch conditions. Experiments on plateau scaling, ruggedness scaling, irrelevant-dimension injection, and seed dependency all reinforce the same conclusion: SHIFT remains effective and robust where conventional methods lose traction.

At the same time, several limitations remain. The current design fixes the maximum compression length statically and supports only integer-valued input domains. Moreover, the compression step can become expensive when many dimensions are present or when rugged regions require repeated detection cycles. Addressing these constraints—through adaptive compression length, symbolic pre-analysis, and broader input-type support—represents a clear path forward.

Overall, SHIFT demonstrates that landscape transformation is a lightweight yet powerful means of rehabilitating search in environments where gradient signals are easily misleading. By reshaping the geometry rather than redesigning the search operators, the method enables simple hill climbing to perform competitively on problems where it would ordinarily stall, offering a practical and extensible tool for SBST and, potentially, for broader classes of discrete optimization.
\bibliographystyle{plainnat}
\bibliography{refs}

\newpage
\appendix

\section{Pilot Study(~\ref{sec:pilot-exp}): Full Synthetic Fitness Landscapes}
\label{appendix:synthetic-landscapes}

\begin{figure}[H]
\centering
\begin{subfigure}{0.48\linewidth}
    \centering
    \includegraphics[width=\linewidth]{work_together/fig/test3/fitness/landscape_needle_1d.pdf}
\end{subfigure}
\hfill
\begin{subfigure}{0.48\linewidth}
    \centering
    \includegraphics[width=\linewidth]{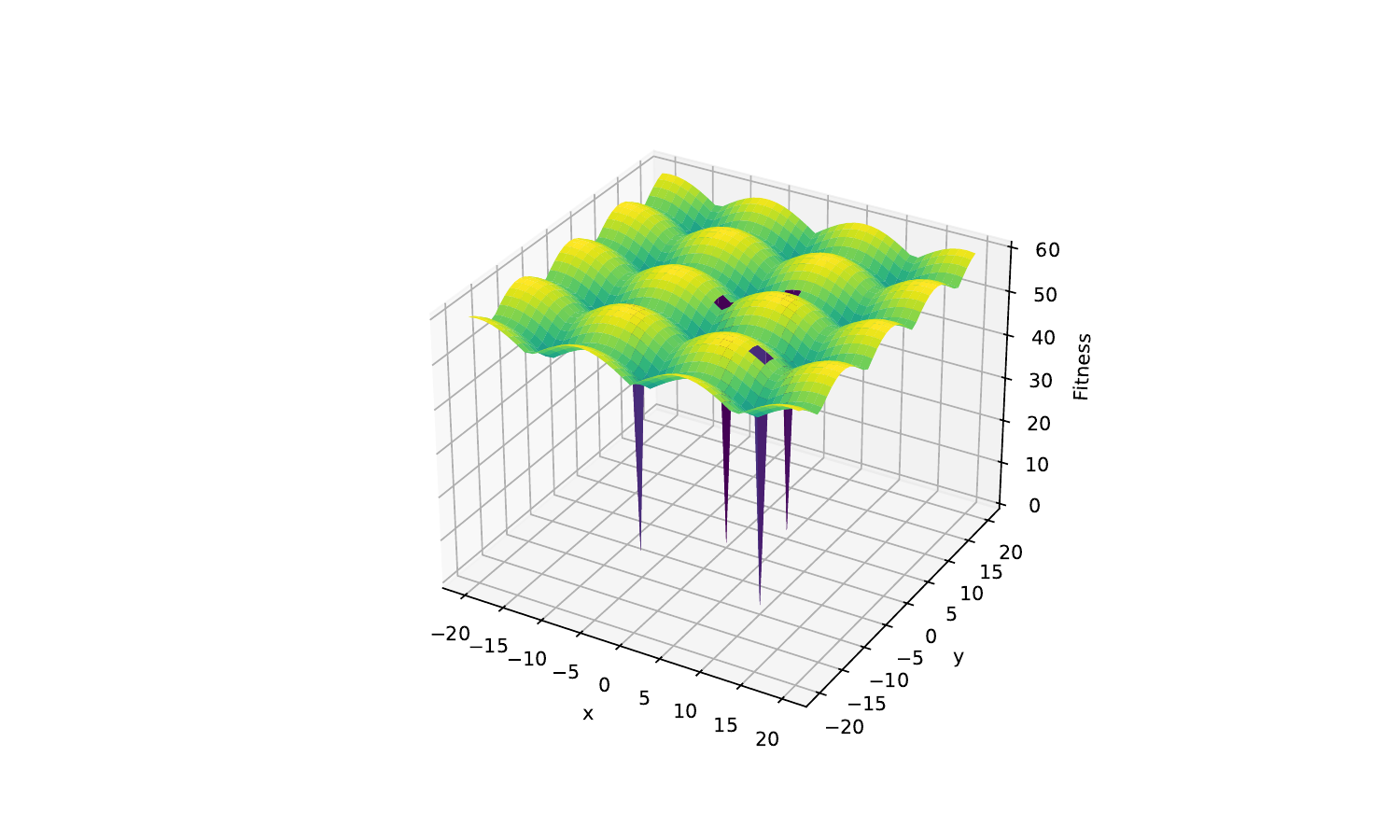}
\end{subfigure}

\caption{Needle landscapes in 1D and 2D.}
\label{fig:appendix-needle}
\end{figure}

\begin{figure}[H]
\centering
\begin{subfigure}{0.48\linewidth}
    \centering
    \includegraphics[width=\linewidth]{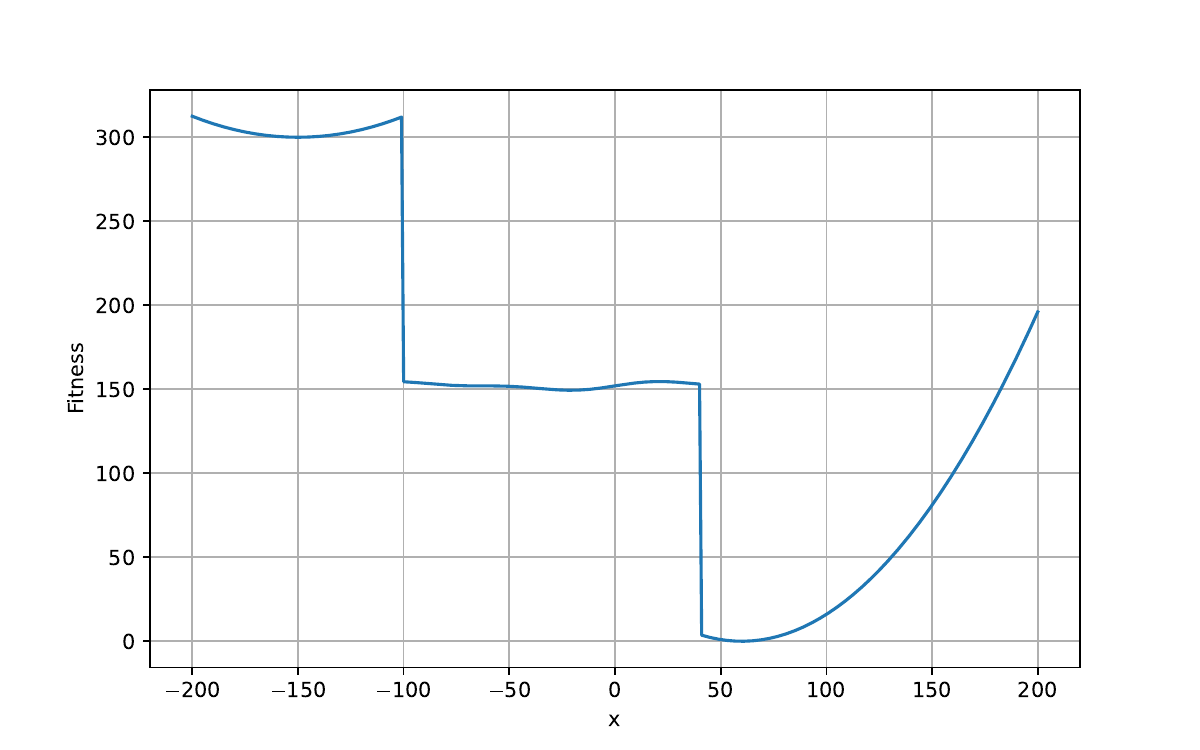}
\end{subfigure}
\hfill
\begin{subfigure}{0.48\linewidth}
    \centering
    \includegraphics[width=\linewidth]{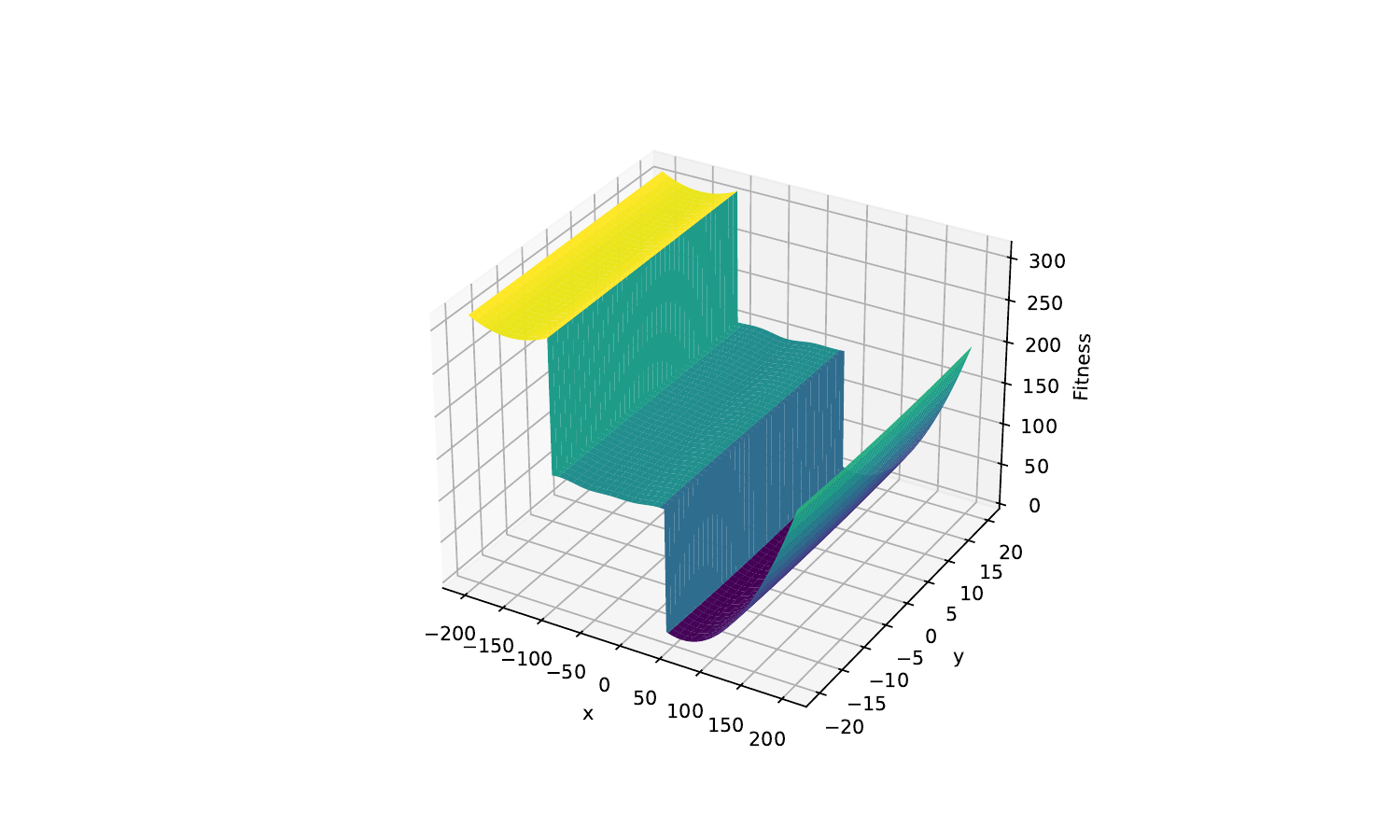}
\end{subfigure}

\caption{Plateau landscapes in 1D and 2D.}
\label{fig:appendix-plateau}
\end{figure}

\begin{figure}[H]
\centering
\begin{subfigure}{0.48\linewidth}
    \centering
    \includegraphics[width=\linewidth]{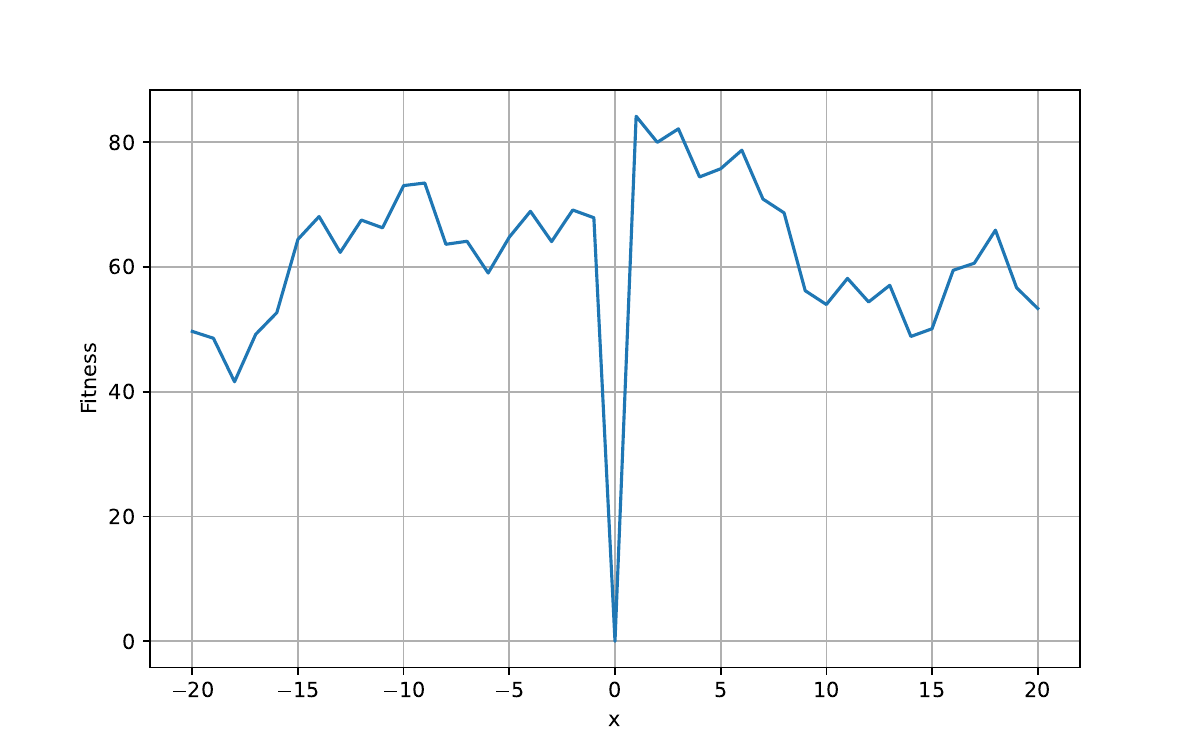}
\end{subfigure}
\hfill
\begin{subfigure}{0.48\linewidth}
    \centering
    \includegraphics[width=\linewidth]{work_together/fig/test3/fitness/landscape_rugged_2d.pdf}
\end{subfigure}

\caption{Rugged landscapes in 1D and 2D.}
\label{fig:appendix-rugged}
\end{figure}

\begin{figure}[H]
\centering
\begin{subfigure}{0.48\linewidth}
    \centering
    \includegraphics[width=\linewidth]{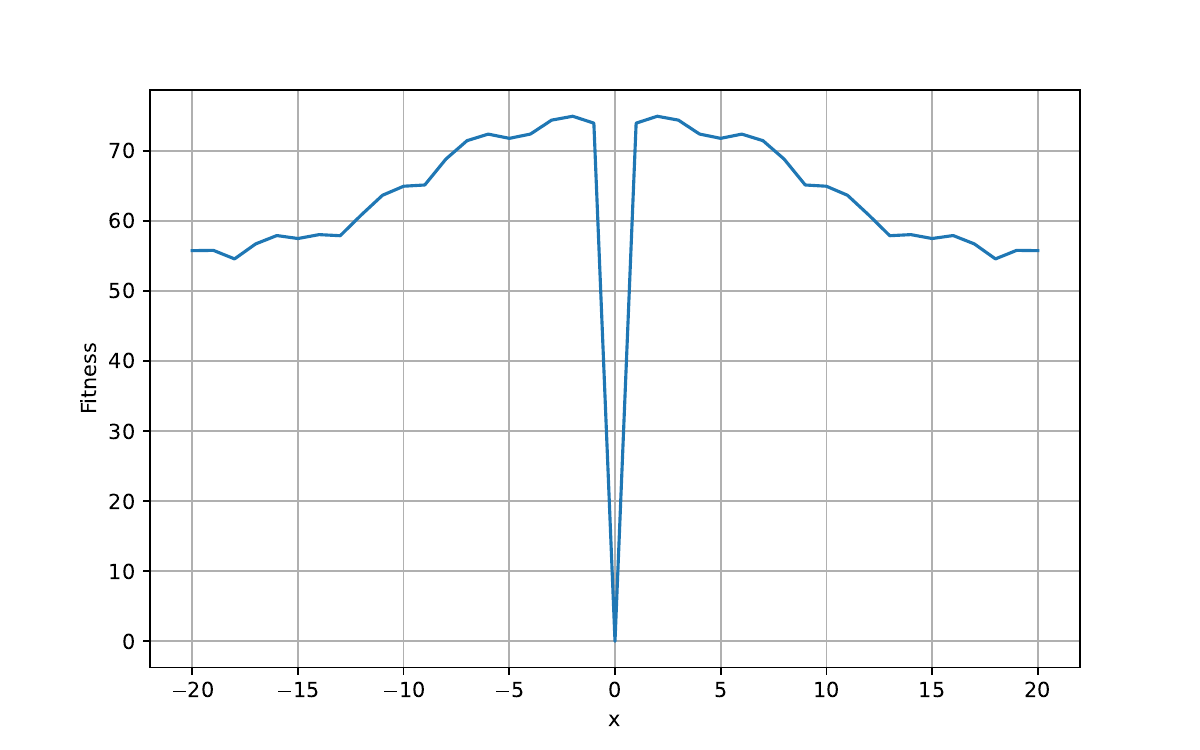}
\end{subfigure}
\hfill
\begin{subfigure}{0.48\linewidth}
    \centering
    \includegraphics[width=\linewidth]{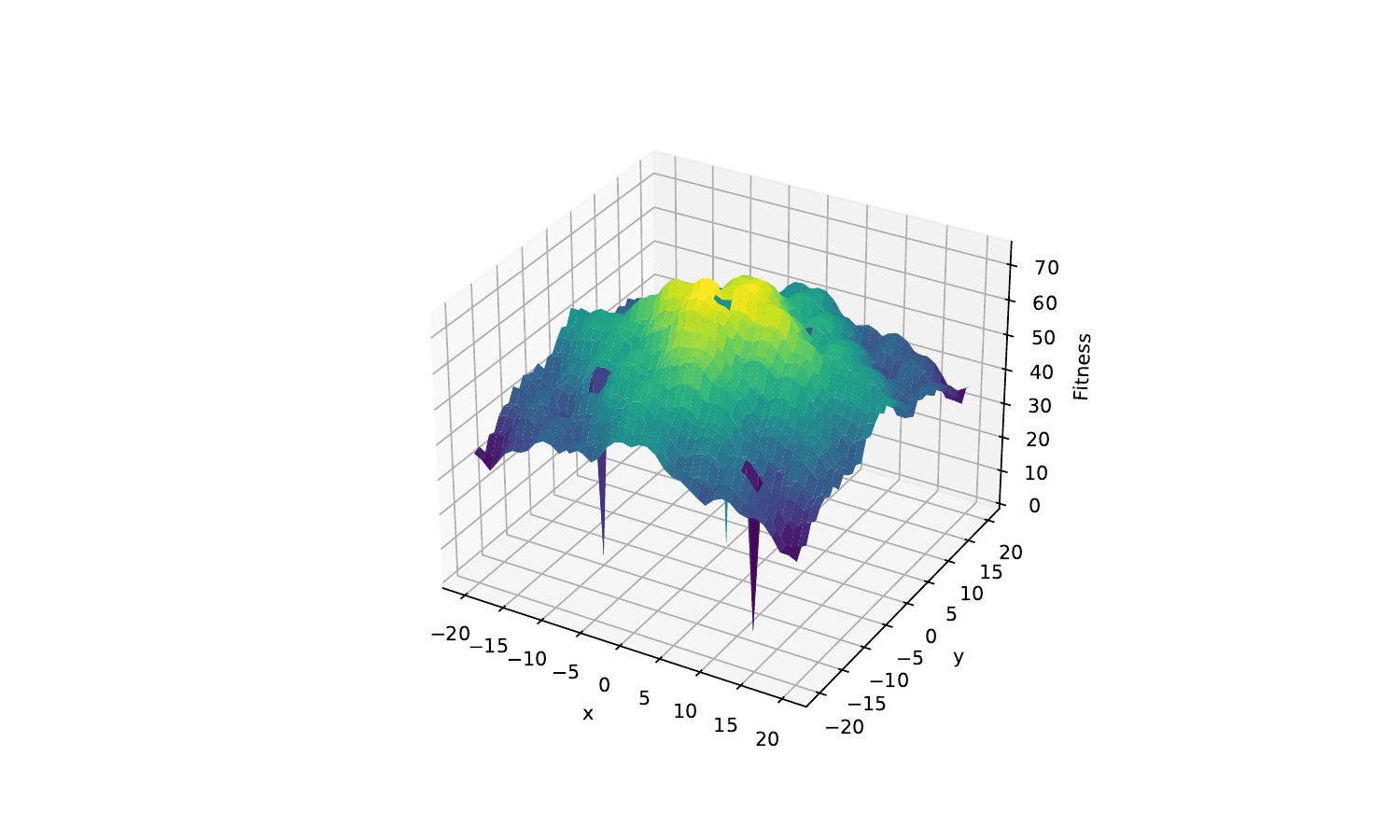}
\end{subfigure}

\caption{Combined landscapes in 1D and 2D.}
\label{fig:appendix-combined}
\end{figure}

\section{Pilot Study(~\ref{sec:pilot-exp}): Hyperparameter Settings}
\begin{table}[H]
\centering
\begin{tabular}{lcc}
\toprule
\textbf{Component} & \textbf{Parameter} & \textbf{Value} \\
\midrule
\multirow{4}{*}{Common} 
    & Time budget per branch $T$ & $20\ \mathrm{s}$ \\
    & Random seed & $42$ \\
    & Input domain $[L,H]$ & auto-detected (AST constants) \\
    & Initialization mode & biased \\
\midrule
\multirow{2}{*}{HC}
    & Max steps per trial & $200$ \\
    & Basin max search & $100$ \\
    & Move set & $\pm 1$ per dimension \\
\midrule
\multirow{3}{*}{HC-SHIFT}
    & Max iterations per dimension & $10$ \\
    & Basin max search per dimension & $100$ \\
    & Active-dimension updates & enabled \\
\midrule
\multirow{4}{*}{GA}
    & Population size $N$ & $1000$ \\
    & Tournament size $k$ & $3$ \\
    & Elite ratio & $0.1$ \\
    & Mutation steps & $\{-3,-2,-1,1,2,3\}$ \\
\bottomrule
\end{tabular}
\caption{Hyperparameter settings for HC, HC-SHIFT, and GA.}
\label{tab:pilot-hyperparams}
\end{table}

\section{Pseudocode for Baseline HC (Section~\ref{sec:hc})}

\begin{algorithm}[H]
\caption{Simple $n$-Dimensional Hill Climbing (HC)}
\label{alg:hc}
\begin{algorithmic}[1]

\Require Fitness function $F$, start point $x$, dimension $n$, max steps $K$, time limit $t_{\max}$.
\Ensure Trajectory of visited points

\State Check time limit; \textbf{if} exceeded, \Return $[(x, \infty)]$
\State $f \gets F(x)$, traj $\gets [(x, f)]$

\For{$k = 1$ to $K$}
    \State Check time limit; \textbf{if} exceeded, \Return traj

    \State $\mathcal{N} \gets \emptyset$ \Comment{axis-aligned neighbors}
    \For{$d = 1$ to $n$}
        \State Check time limit; \textbf{if} exceeded, \Return traj
        \State Add $(x_1,\ldots,x_d - 1,\ldots,x_n)$ to $\mathcal{N}$
        
        \State Check time limit; \textbf{if} exceeded, \Return traj
        \State Add $(x_1,\ldots,x_d + 1,\ldots,x_n)$ to $\mathcal{N}$
    \EndFor

    \State Evaluate $F(y)$ for all $y \in \mathcal{N}$
    \State $(x', f') \gets \arg\min_{y \in \mathcal{N}} F(y)$ \Comment{steepest descent}

    \If{$f' < f$}
        \State $x \gets x'$, $f \gets f'$
        \State Append $(x, f)$ to traj
    \Else
        \State \Return traj \Comment{local minimum reached}
    \EndIf
\EndFor

\State \Return traj

\end{algorithmic}
\end{algorithm}

\section{Pseudocode for HC-SHIFT (Section~\ref{sec:compression})}
\begin{algorithm}[H]
\caption{HC-SHIFT: Main Algorithm}
\label{alg:hc-shift-main}
\begin{algorithmic}[1]

\Require Fitness function $F$, start point $x$, dimension $n$, max iterations $k_{\max}$, time limit $t_{\max}$, patience $p$.
\State $x \gets$ initial point, $f \gets F(x)$
\State $\mathcal{A} \gets \{0,\dots,n-1\}$ \Comment{active dimensions}
\State $\sigma_d \gets 0$ for all $d$ \Comment{stagnation counters}
\State $\mathcal{CM} \gets$ CompressionManagerND($n$)

\If{$f < \epsilon$} \Return $x$ \Comment{early success}
\EndIf

\For{iteration $i = 1$ to $k_{\max}$}
    \If{$\mathcal{A} = \emptyset$ \textbf{or} time $> t_{\max}$}
        \State \Return best $x$
    \EndIf
    
    \If{$f < \epsilon$} \Return $x$
    \EndIf

    \State $x, f \gets$ \textsc{HillClimbLoop}($x, f, \mathcal{A}, \mathcal{CM}, \sigma, p$) \Comment{Alg. \ref{alg:hill-climb-loop}}
    
    \If{$f < \epsilon$} \Return $x$
    \EndIf

    \State $\mathcal{B} \gets$ \textsc{DetectBasins}($x, \mathcal{A}, \mathcal{CM}$) \Comment{Alg. \ref{alg:basin-detection}}
    
    \If{$\mathcal{B} = \emptyset$}
        \State \Return best $x$ \Comment{no compressible basin}
    \EndIf

    \State $x, f \gets$ \textsc{SelectRestartPoint}($x, \mathcal{B}$) \Comment{Alg. \ref{alg:restart-selection}}
    
    \If{$f < \epsilon$} \Return $x$
    \EndIf
\EndFor

\State \Return best $x$

\end{algorithmic}
\end{algorithm}

\begin{algorithm}[H]
\caption{\textsc{DetectBasins}: Basin Detection and Compression Update}
\label{alg:basin-detection}
\begin{algorithmic}[1]

\Require Current point $x$, active dimensions $\mathcal{A}$, compression manager $\mathcal{CM}$.
\Ensure Basin set $\mathcal{B}$ (each element: $(d, (b_\text{start}, b_\text{len}))$)

\State $\mathcal{B} \gets \emptyset$

\For{each dimension $d \in \mathcal{A}$}
    \State Check time limit; \textbf{if} exceeded, \Return $\mathcal{B}$
    
    \State $B_d \gets$ \textsc{Detect1DBasin}($x, d$) \Comment{bidirectional search}
    
    \If{$B_d \neq \emptyset$}
        \State Get fixed coordinates: $\text{fixed} \gets (x_i : i \neq d)$
        \State $\mathcal{CM}$.\textsc{UpdateDimension}($d$, fixed, $B_d$)
        \State Add $(d, B_d)$ to $\mathcal{B}$
    \EndIf
\EndFor

\State \Return $\mathcal{B}$

\end{algorithmic}
\end{algorithm}

\begin{algorithm}[H]
\caption{\textsc{SelectRestartPoint}: Basin Boundary Evaluation}
\label{alg:restart-selection}
\begin{algorithmic}[1]

\Require Current point $x$, basin set $\mathcal{B}$.
\Ensure Restart point $x_\text{restart}$, fitness $f_\text{restart}$

\State $\mathcal{R} \gets \emptyset$ \Comment{restart candidates}

\For{each $(d, (b_\text{start}, b_\text{len})) \in \mathcal{B}$}
    \State Check time limit; \textbf{if} exceeded, \Return $(x, F(x))$
    
    \State $b_\text{end} \gets b_\text{start} + b_\text{len} - 1$
    
    \State Create point $y^-$ by setting $x_d \gets b_\text{start} - 1$
    \State Add $(y^-, F(y^-))$ to $\mathcal{R}$
    
    \State Create point $y^+$ by setting $x_d \gets b_\text{end} + 1$
    \State Add $(y^+, F(y^+))$ to $\mathcal{R}$
\EndFor

\State $(x_\text{restart}, f_\text{restart}) \gets \arg\min_{(y, f_y) \in \mathcal{R}} f_y$

\State \Return $(x_\text{restart}, f_\text{restart})$

\end{algorithmic}
\end{algorithm}

\begin{algorithm}[H]
\caption{\textsc{HillClimbLoop}: Compressed-Space Hill Climbing}
\label{alg:hill-climb-loop}
\begin{algorithmic}[1]

\Require Current point $x$, fitness $f$, active dims $\mathcal{A}$, compression manager $\mathcal{CM}$, stagnation $\sigma$, patience $p$.
\Ensure Updated point $x$, fitness $f$

\State $\text{step\_count} \gets 0$
\Repeat
    \State Check time limit; \textbf{if} exceeded, \Return $(x, f)$
    
    \State $\mathcal{N} \gets$ \textsc{GenerateNeighbors}($x, \mathcal{A}, \mathcal{CM}$) \Comment{Alg. \ref{alg:generate-neighbors}}
    
    \State $x_\text{best}, f_\text{best} \gets x, f$
    \State $\mathcal{M} \gets \emptyset$ \Comment{meaningful dimensions}
    
    \For{each $(y, f_y, \mathcal{D}_y) \in \mathcal{N}$}
        \If{$f_y < f_\text{best}$}
            \State $x_\text{best}, f_\text{best} \gets y, f_y$
        \EndIf
        \If{$f_y \neq f$}
            \State $\mathcal{M} \gets \mathcal{M} \cup \mathcal{D}_y$
        \EndIf
    \EndFor

    \State Update stagnation: $\sigma_d \gets \begin{cases} 0 & \text{if } d \in \mathcal{M} \\ \sigma_d + 1 & \text{otherwise} \end{cases}$
    
    \State Deactivate: $\mathcal{A} \gets \mathcal{A} \setminus \{d : \sigma_d \geq p\}$

    \If{$f_\text{best} < f$}
        \State $x \gets x_\text{best}$, $f \gets f_\text{best}$
        \State $\text{step\_count} \gets \text{step\_count} + 1$
    \Else
        \State \textbf{break} \Comment{local minimum}
    \EndIf
\Until{$\text{step\_count} \geq \text{max\_steps}$}

\State \Return $(x, f)$

\end{algorithmic}
\end{algorithm}

\begin{algorithm}[H]
\caption{\textsc{GenerateNeighbors}: Compressed Neighbor Generation}
\label{alg:generate-neighbors}
\begin{algorithmic}[1]

\Require Current point $x$, active dimensions $\mathcal{A}$, compression manager $\mathcal{CM}$.
\Ensure Neighbor set $\mathcal{N}$ (each element: $(y, f_y, \mathcal{D}_y)$)

\State $\mathcal{N} \gets \emptyset$

\For{each $d \in \mathcal{A}$} \Comment{Axis-aligned neighbors}
    \State Get compression system $w_d$ for dimension $d$ (if exists)
    \If{$w_d$ exists}
        \State $z_d \gets w_d(x_d)$
        \State $n^- \gets w_d^{-1}(z_d-1)$, $n^+ \gets w_d^{-1}(z_d+1)$
    \Else
        \State $n^- \gets x_d - 1$, $n^+ \gets x_d + 1$
    \EndIf
    \State Create neighbors $y^-$ and $y^+$ by modifying $x$ at dimension $d$
    \State Add $(y^-, F(y^-), \{d\})$ and $(y^+, F(y^+), \{d\})$ to $\mathcal{N}$
\EndFor

\If{$|\mathcal{A}| \geq 2$} \Comment{Diagonal neighbors}
    \For{each pair $(d_1, d_2)$ in combinations of $\mathcal{A}$}
        \State Get compression systems $w_{d_1}$, $w_{d_2}$ (if exist)
        \State Compute neighbor values for $d_1$: $\{n_1^-, n_1^+\}$ (compressed or not)
        \State Compute neighbor values for $d_2$: $\{n_2^-, n_2^+\}$ (compressed or not)
        \For{$v_1 \in \{n_1^-, n_1^+\}$, $v_2 \in \{n_2^-, n_2^+\}$}
            \State Create diagonal neighbor $y$ by modifying $x$ at dims $d_1, d_2$
            \State Add $(y, F(y), \{d_1, d_2\})$ to $\mathcal{N}$
        \EndFor
    \EndFor
\EndIf

\State \Return $\mathcal{N}$

\end{algorithmic}
\end{algorithm}

\newpage
\section{Complete Benchmark Results (~\ref{sec:exp1})}
\label{appendix:benchmark_res}
\begin{table*}[!htbp]
\centering
\small
\begin{tabular}{l|cc|cc|cc}
\toprule
\textbf{Benchmark}
& \multicolumn{2}{c|}{\textbf{HC-SHIFT}}
& \multicolumn{2}{c|}{\textbf{HC}}
& \multicolumn{2}{c}{\textbf{GA}} \\
& Avg. Trials & Avg. Time (s) & Avg. Trials & Avg. Time (s) & Avg. Trials & Avg. Time (s) \\
\midrule
arbitrary1 & 1.00 & 0.008 & 1.00 & 0.006 & 18.42 & 0.188 \\
arbitrary2 & 1.00 & 0.448 & 1.29 & 0.004 & 5.86 & 0.149 \\
arbitrary3 & 1.00 & 0.009 & 1.00 & 0.008 & 20.00 & 0.141 \\
arbitrary4 & 1.00 & 0.006 & 1.38 & 0.006 & 9.00 & 0.162 \\
arbitrary5 & 1.00 & 0.007 & 1.00 & 0.007 & 6.50 & 0.119 \\
arbitrary6 & 1.17 & 1.243 & 1.33 & 0.004 & 1.67 & 0.106 \\
arbitrary7 & 1.00 & 0.634 & 1.67 & 0.007 & 4.00 & 0.163 \\
arbitrary8 & 17.67 & 3.340 & 4583.50 & 3.337 & 4488.67 & 3.496 \\
arbitrary9 & 1.00 & 0.011 & 1.11 & 0.005 & 2.00 & 0.143 \\
arbitrary10 & 1.00 & 0.013 & 1.12 & 0.006 & 207.75 & 0.147 \\
collatz\_step & 1.00 & 0.015 & 6.00 & 0.008 & 9.67 & 0.086 \\
combined1 & 1.00 & 0.356 & 2.21 & 0.493 & 4417.86 & 0.476 \\
combined2 & 8.19 & 16.738 & 38292.06 & 12.536 & 247545.00 & 12.599 \\
count\_divisor\_1 & 3.08 & 5.108 & 9350.92 & 5.004 & 6038.17 & 5.103 \\
count\_divisor\_2 & 2.58 & 4.145 & 8312.83 & 6.674 & 28765.00 & 6.829 \\
derivative\_quadratic & 70.00 & 10.008 & 264995.00 & 10.002 & 7139.50 & 10.051 \\
digit\_sum & 1.00 & 0.006 & 1.50 & 0.006 & 9.80 & 0.082 \\
ex1 & 1.00 & 0.008 & 1.00 & 0.007 & 23.62 & 0.190 \\
ex2 & 1.00 & 0.019 & 1.00 & 0.019 & 2.33 & 0.133 \\
ex3 & 1.00 & 0.018 & 1.07 & 0.010 & 31.07 & 0.319 \\
ex4 & 1.00 & 0.047 & 1.00 & 0.024 & 5202.50 & 0.264 \\
ex5 & 2.33 & 2.161 & 5.00 & 0.227 & 69590.17 & 2.733 \\
ex6 & 1.00 & 0.009 & 2.17 & 0.010 & 64.83 & 0.041 \\
ex7 & 1.00 & 0.003 & 1.00 & 0.004 & 5.00 & 0.037 \\
mixed\_case & 1.00 & 1.672 & 13.67 & 5.782 & 141325.50 & 8.451 \\
needle1 & 1.00 & 0.293 & 1.80 & 0.824 & 2.60 & 0.279 \\
needle2 & 51.50 & 10.007 & 16628.75 & 10.002 & 259999.50 & 10.102 \\
needle\_case & 29.75 & 10.012 & 131752.50 & 10.002 & 176919.75 & 10.043 \\
parallel\_test & 1.12 & 0.288 & 4.31 & 0.010 & 29.58 & 0.057 \\
plateau1 & 1.00 & 0.101 & 1.00 & 0.037 & 28.38 & 0.163 \\
plateau2 & 3.75 & 0.572 & 106216.25 & 5.005 & 233190.88 & 5.056 \\
plateau3 & 1.00 & 0.053 & 1.00 & 0.021 & 129301.81 & 6.435 \\
plateau\_case & 8.50 & 4.167 & 86845.62 & 10.057 & 267704.38 & 12.584 \\
prime\_check & 1.00 & 0.024 & 2.75 & 0.005 & 5.58 & 0.052 \\
rugged1 & 1.00 & 0.028 & 1.00 & 0.017 & 2.06 & 0.261 \\
rugged2 & 53.25 & 10.010 & 17306.00 & 10.003 & 295836.75 & 10.063 \\
rugged\_case & 1.00 & 1.143 & 242500.00 & 10.002 & 16167.50 & 10.071 \\
triangle & 1.83 & 2.524 & 8.08 & 0.012 & 20.00 & 0.220 \\
\bottomrule
\end{tabular}
\caption{Per-benchmark average number of trials and total time (in seconds) for biased initialization, aggregated over all branches in each file.}
\label{tab:benchmark_biased}
\end{table*}

\end{document}